\documentclass[journal]{IEEEtran}
\usepackage{amsmath}
\allowdisplaybreaks[4]
\usepackage{amsthm}
\usepackage{amssymb}
\usepackage{verbatim}
\usepackage[linesnumbered,ruled,vlined]{algorithm2e}
\usepackage{graphicx}
\usepackage{multirow}
\usepackage{multicol}
\usepackage{booktabs}
\usepackage{makecell}
\usepackage{paralist}
\usepackage{url}
\usepackage{color}
\usepackage[colorlinks,
 linkcolor=black,
 urlcolor=black,
 anchorcolor=black,
 citecolor=black]{hyperref}
\usepackage{cite}
\usepackage{xcolor,colortbl}
\usepackage{float}
\usepackage{subfigure}
\newtheorem{lemma}{Lemma}
\usepackage{lipsum}

\begin{document}

\title{Joint Network Function Placement and Routing Optimization in Dynamic Software-defined Satellite-Terrestrial Integrated Networks}

\author{{Shuo~Yuan,~\IEEEmembership{Member,~IEEE},
			Yaohua~Sun,
			and~Mugen~Peng,~\IEEEmembership{Fellow,~IEEE}}
	\thanks{This work was supported in part by the Beijing Municipal Science and Technology Project under Grant Z211100004421017, in part by the National Natural Science Foundation of China under Grants 62371071 and 62001053, and in part by the Young Elite Scientists Sponsorship Program by the China Association for Science and Technology under Grant 2021QNRC001.
	}
	\thanks{Shuo~Yuan (yuanshuo@bupt.edu.cn), Yaohua~Sun (sunyaohua@bupt.edu.cn), and Mugen~Peng (pmg@bupt.edu.cn) are with the State Key Laboratory of Networking and Switching Technology, Beijing University of Posts and Telecommunications, Beijing 100876, China. 
	(\emph{Corresponding author: Yaohua Sun})
	}
}


\maketitle

\begin{abstract}
	Software-defined satellite-terrestrial integrated networks (SDSTNs) are seen as a promising paradigm for achieving high resource flexibility and global communication coverage. However, low latency service provisioning is still challenging due to the fast variation of network topology and limited onboard resource at low earth orbit satellites. To address this issue, we study service provisioning in SDSTNs via joint optimization of virtual network function (VNF) placement and routing planning with network dynamics characterized by a time-evolving graph. Aiming at minimizing average service latency, the corresponding problem is formulated as an integer nonlinear programming under resource, VNF deployment, and time-slotted flow constraints. Since exhaustive search is intractable, we transform the primary problem into an integer linear programming by involving auxiliary variables and then propose a Benders decomposition based branch-and-cut (BDBC) algorithm. Towards practical use, a time expansion-based decoupled greedy (TEDG) algorithm is further designed with rigorous complexity analysis. Extensive experiments demonstrate the optimality of BDBC algorithm and the low complexity of TEDG algorithm. Meanwhile, it is indicated that they can improve the number of completed services within a configuration period by up to 58\% and reduce the average service latency by up to 17\% compared to baseline schemes.
\end{abstract}


\IEEEpeerreviewmaketitle

\section{Introduction}
\IEEEPARstart{S}{atellite} networks are able to provide communication services for remote areas, e.g. deserts and oceans, that lack terrestrial communication infrastructure.
Owing to this merit, the paradigm of satellite-terrestrial integrated networks (STNs) is considered as the key to achieving global coverage and ubiquitous connectivity in the sixth generation (6G) era \cite{chen2020visionrequirements}.
Compared with geostationary earth orbit satellites and medium earth orbit satellites, low earth orbit (LEO) satellites have a lower altitude with the range of 300-2000 $km$ above the earth surface, which means reduced propagation latency, better signal quality, and lower launch cost \cite{kodheli2021satellitecommunications}.
Representative commercial LEO satellite programs include Starlink and OneWeb \cite{delportillo2019technicalcomparison}.

In conventional practice, LEO satellite payloads have generally been designed for dedicated services, such as remote sensing or navigation \cite{braun2022satellitecommunications}.
Consequently, the resource of the payload, such as communication and computation resource, cannot be shared by other services, even when the payload is idle.
This limitation leads to low resource utilization.
To address this issue, some studies have introduced network function virtualization (NFV) technology into satellite payload design \cite{qin2023serviceawareresource, yang2023spaceinformation,shi2019crossdomainsdn}, which substitutes hardware middleboxes with software modules known as virtual network functions (VNFs).
The deployment of VNFs offers more flexibility compared to the manual installation of hardware middleboxes, as it involves initiating a VNF instance using virtualized resources through a prebuilt file \cite{ma2022mobilityawaredelaysensitive}.
Furthermore, by incorporating software-defined network (SDN) to separate the control plane and the data plane of STNs \cite{yuan2022softwaredefined}, software-defined STNs (SDSTNs) are formed, featuring flexible network management and system performance control.

\subsection{Related Works}

Different from regular resource allocation strategies in STNs, which typically focus on aspects such as beam scheduling, power allocation, and subchannel assignment to meet user data rate demands \cite{wang2022jointoptimization} or optimize system throughput \cite{di2019ultradenseleoa}, SDSTNs offer more flexible resource management, such as deploying VNFs and planning service routing paths on-demand, to improve service quality or network profit \cite{yuan2022softwaredefined}.
These regular resource allocation strategies cannot be directly applied to the joint VNF placement and routing planning in STNs due to the following challenges.
First, the relative motion among satellites leads to dynamic network topologies and variable resource availability, making it difficult to characterize heterogeneous resources in dynamic STNs.
The existing works often model the network topology in a static manner, which limits their applicability in dynamic scenarios where topology changes are frequent \cite{di2019ultradenseleoa, tang2021computationoffloading, cheng2019spaceaerialassistedcomputing}.
Second, VNF placement in STNs is closely coupled with service routing path planning.
Although there are some studies on computing offloading that consider the selection of offloading paths, these optional offloading paths are often deterministic and limited in number \cite{tang2021computationoffloading,cheng2019spaceaerialassistedcomputing,qiu2019deepqlearning}.
However, in the context of VNF placement, the service path must traverse all relevant VNFs, which may be placed on different network nodes.
Consequently, it becomes critically important to accurately model flow conservation constraints in STNs.
Moreover, the limited satellite payloads impose constraints on onboard resource, including computation units and bandwidth, which are scarce and must be carefully allocated and scheduled to ensure high resource utilization.

For routing planning in STNs, an energy-efficient satellite routing algorithm is designed in \cite{yang2016energyefficientrouting} where the number of hops and the load of links are considered.
To avoid satellite congestion and reduce queuing latency, literature \cite{zhu2017softwaredefined} proposes a software-defined routing algorithm based on a distributed master-slave SDN controller architecture.
In \cite{tang2020dynamicallyadaptive}, a cross-domain adaptive cooperation routing scheme for SDN-based layered STNs is presented to balance the load among nodes and maximize the network throughput.
Moreover, some other works have investigated the joint optimization of VNF placement and routing planning.
In \cite{wang2020sfcbasedservice}, a sharing-based VNF placement and service data routing problem is studied to minimize resource cost, and a greedy algorithm is proposed by considering the differential features of satellite and ground nodes.
To lower cross-domain service latency, a multi-domain VNF mapping orchestration framework with distributed controllers is proposed in \cite{li2018horizontalbasedorchestration}.
In \cite{gao2022dynamicresource}, the VNF placement problem in satellite edge clouds is formulated as an integer non-linear programming problem to jointly minimize the network bandwidth cost and service end-to-end latency and a heuristic distributed VNF placement algorithm is developed.

However, the above works assume that the topologies of SDSTNs are static, which are inconsistent with practical situations.
Although some works \cite{ma2022mobilityawaredelaysensitive,chen2019mobilityawareservice} incorporate user mobility in VNF placement optimization and introduce user dynamics into network modeling, the target scenarios of these works are primarily focused on terrestrial networks.
In such contexts, the paths between network nodes are generally considered to be predetermined and stable, which reduces the necessity for further consideration of routing planning.
However, in SDSTNs, the availability of inter-satellite links (ISLs), satellite-terrestrial links (STLs), and computation resource vary over time, leading to a highly dynamic network structure. These dynamic natures of SDSTNs necessitate being captured precisely for the exploration of joint VNF placement and routing planning.

To characterize the dynamic network topology and availability of heterogeneous resource, some researchers have adopted a time-evolving graph (TEG) when modeling concerned systems, such as \cite{zhou2017missionaware,zhou2019distributionallyrobust,jia2021vnfbasedservice,gao2022virtualnetwork}.
In \cite{jia2021vnfbasedservice}, the authors formulate the joint VNF placement and routing planning in dynamic SDSTNs as an integer linear programming problem based on TEG with the goal of minimizing communication resource consumption, and a branch-and-price algorithm is proposed.
In \cite{gao2022virtualnetwork}, the authors use a TEG to model the VNF deployment problem as a potential game and then propose a decentralized algorithm to minimize the network payoff, which includes bandwidth and energy consumption costs.
In addition, the authors in \cite{li2022costawaredynamic} propose a Tabu search-based VNF placement algorithm for randomly arrived services in STNs.
This algorithm dynamically adjusts the VNF placement strategy to maximize the service provider's profit while addressing the end-to-end latency requirement of newly arrived service requests.
To minimize the transmission cost of computation tasks in time-varying STNs, the authors in \cite{li2022processingwhiletransmittingcostminimized} formulate a joint service routing and VNF placement problem and then propose a classic shortest path-based solution.
However, there have been few efforts toward reducing service end-to-end latency in the context of VNF placement and routing planning under multi-time slot dynamic STN topologies, and there is a lack of research on modeling time-slotted service flow conservation constraints.

\subsection{Motivations and Contributions}

Based on the above literature review and observations, the following limitations can be identified for VNF placement and routing planning in existing works concerning dynamic STNs, i.e.,
\begin{enumerate}
	\item
	      Although some existing works, such as \cite{li2022processingwhiletransmittingcostminimized, jia2021vnfbasedservice}, have explored VNF placement in dynamic STN topologies, they focus on optimizing resource efficiency and network profit.
	      However, there remains a lack of research addressing the reduction of end-to-end service latency under these dynamic topologies.
	\item
	      Concerning routing planning coupled with VNF placement, existing works, such as \cite{li2022costawaredynamic,wang2020sfcbasedservice, gao2022dynamicresource, li2018horizontalbasedorchestration, chen2019mobilityawareservice, gao2022virtualnetwork, li2022processingwhiletransmittingcostminimized}, do not consider routing across multiple time slots in STNs. Consequently, there is a scarcity of research examining service flow conservation under dynamic network topologies.
	      Furthermore, the undetermined service completion time poses challenges for modeling time-slotted service flow conservation.
	\item
	      Exclusive allocation strategies over large periods are commonly employed in existing works, which often limit the reallocation of idle resource after completing a service and resulting in low resource utilization efficiency.
	\item
	      The joint VNF placement and routing planning problems are generally NP-hard \cite{jia2021vnfbasedservice,gao2022virtualnetwork,li2022costawaredynamic}, making it challenging to obtain an optimal solution with polynomial complexity.
	      Due to the dynamic nature of STN topologies, additional challenges arise, including larger variable and constraint sizes, as well as nonlinear flow conservation constraints, when devising effective solutions for such joint problems.
\end{enumerate}

Motivated by these observations, in this paper, the communication resource of each communication link and computation resource at each network node are allocated to each service on a per-time-slot basis.
In addition, the joint VNF placement and routing planning problem is formulated to minimize the average end-to-end service latency under dynamic flow constraints and resource constraints of LEO satellites and ground nodes.
Since the formulated problem is NP-hard, it is intractable to obtain optimal solutions through exhaustive search.
Therefore, we utilize Benders decomposition (BD) to decompose the primal problem into a master problem and a subproblem.
Benders cuts are generated by solving the subproblem, which helps tighten the solution space of the master problem iteratively.
In addition, to obtain a low complexity solution for practical use, a time expansion-based heuristic algorithm is further proposed.

\begin{table}[t]
	\caption{List of notations}
	\label{tab:notations}
	\centering
	\begin{tabular}{c|l}
		\toprule[1pt]
		\textbf{Notation}   & \multicolumn{1}{c}{\textbf{Description}}                    \\ \midrule[1pt]
		$\cal G$            & The time-evolving graph of SDSTNs                           \\
		${\cal N}, N$       & The set and number of network nodes                         \\
		${\cal N}_s, N_s$   & The set and number of satellites                            \\
		${\cal N}_g, N_g$   & The set and number of ground stations                       \\
		${\cal N}_u, N_u$   & The set and number of ground users                          \\
		${\cal L}, L$       & The set and number of links                                 \\
		${\cal L}_c, L_c$   & The set and number of communication links                   \\
		${\cal L}_s, L_s$   & The set and number of stay links                            \\
		${R_{nm}^t}$        & Available transmission rate of link $nm$ in time slot $t$   \\
		$\varphi _{nm}^t$   & Parameter indicates whether link $nm$ is active             \\
		${d_{nm}^t}$        & Distance of link $nm$ in time slot $t$                      \\
		${\cal T}$, $T$     & Time horizon and total number of time slots                 \\
		$t$, $\tau$         & The $t$-th time slot in ${\cal T}$ and time slot length     \\
		${\cal Q}, Q$       & The set and number of services                              \\
		${\cal I}_q, I_q$   & The set and number of VNFs of service $q$                   \\
		${{\cal L}}_q, L_q$ & The set and number of virtual links of service $q$          \\
		$s_q$               & The source node of service $q$                              \\
		$e_q$               & The destination node of service $q$                         \\
		$c_q$               & The computation resource requirement of service $q$         \\
		$\delta_q$          & The initial data size of service $q$                        \\
		${\beta _i^q}$      & \makecell[l]{Computation resource consumption of $i$-th VNF \\of service $q$ for deployment}                     \\
		${\varepsilon}$     & The computation resource for 1-bit data processing          \\
		${x_n^{q,i}}$       & \makecell[l]{A binary variable that indicates whether       \\$i$-th VNF of service $q$ is deployed on node $n$}            \\
		${y_{nm}^{q,i,t}}$  & \makecell[l]{A binary variable that indicates whether       \\$i$-th virtual link of service $q$ passes through link $nm$} \\
		\bottomrule[1pt]
	\end{tabular}
\end{table}

The contributions of this paper are summarized as follows.

\begin{itemize}
	\item
	      The joint VNF placement and routing planning problem is considered with the objective of minimizing average end-to-end service latency in SDSTNs, where TEG is leveraged to model the dynamics of network topology and availability of communication and computation resource. Moreover, under the assumption that network resource can be allocated among services on a per-time-slot basis, enhanced flow conservation constraints are derived.
	\item
	      Since the primal problem is an integer non-linear programming (INLP), which is intractable to obtain optimal solutions by exhaustive search, the primal problem is first converted into an integer linear programming (ILP) by involving multiple auxiliary variables, and then an iterative Benders decomposition based branch-and-cut (BDBC) algorithm is proposed to accelerate the acquisition of optimal solutions. To further lower the computational complexity for practical use, a time expansion-based decoupled greedy (TEDG) algorithm is designed.
	\item
	      We rigorously analyze the computational complexity of the TEDG algorithm with regard to the number of services, operation time slots, and the number of network nodes.
	      Extensive simulations are conducted to evaluate the performance of both proposals, in terms of optimality and execution time. Moreover, it is shown that they can improve the performance of average service latency and meanwhile contribute to much more completed services within a configuration period compared to baseline schemes.
\end{itemize}

The remainder of this paper is organized as follows.
System model and problem formulation are described in Section \ref{sec:systemModel}.
Problem decomposition and optimal algorithm design are provided in Section \ref{sec:bcbd4exactSolution}, and the TEDG algorithm is proposed in Section \ref{sec:greedySolution}.
The performance evaluation and corresponding analyses are presented in Section \ref{sec:performanceEvaluations}.
Finally, Section \ref{sec:conclusion} concludes this paper.
For convenience, some important notations are listed in Table \ref{tab:notations}.

\begin{figure}[t]
	\centering
	\begin{minipage}[t]{0.9\linewidth}
		\centering
		\includegraphics[width=\textwidth]{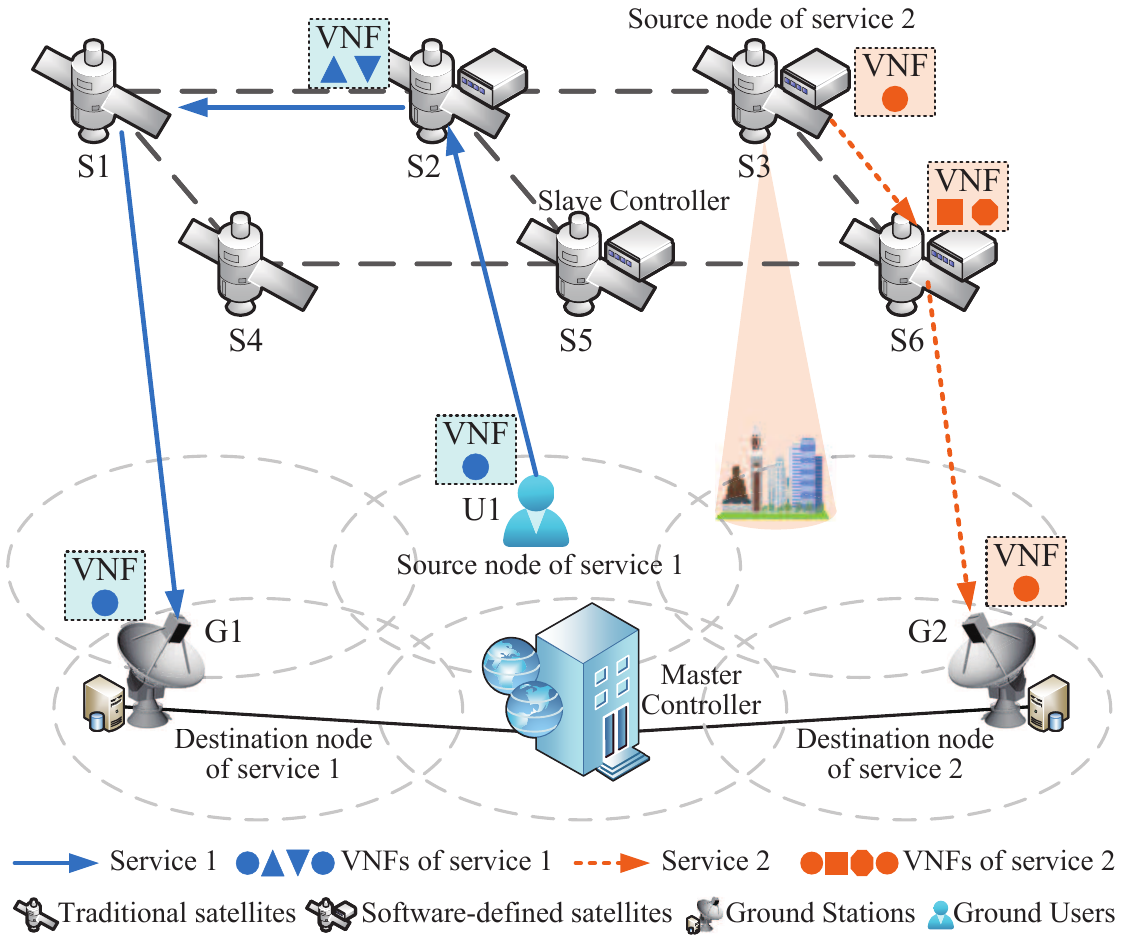}
		\vspace{-0.2in}
		\caption{An example of SDSTN topology.}
		\label{fig:des_topology}
	\end{minipage}%
	\vspace{0.15in}
	\hfill%
	\begin{minipage}[t]{0.9\linewidth}
		\centering
		\includegraphics[width=\textwidth]{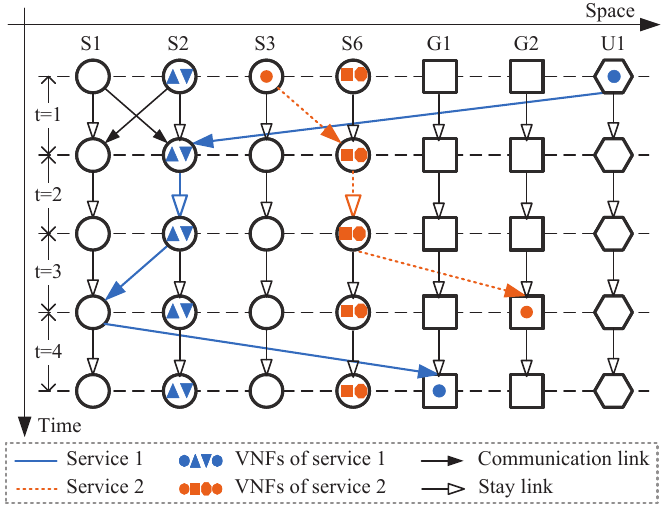}
		\vspace{-0.2in}
		\caption{An example of the TEG model for SDSTNs.}
		\label{fig:des_TEGmodel}
	\end{minipage}
\end{figure}

\section{System Model and Problem Formulation}
\label{sec:systemModel}
Fig. \ref{fig:des_topology} illustrates an example topology of SDSTNs, which consists of a set of LEO satellites, ground stations, and ground users.
Ground stations and ground users are able to establish connections with satellites within line of sight using STLs.
Meanwhile, ground stations with computation resource are interconnected via optical fiber.
LEO satellites can be categorized as traditional satellites and software-defined satellites equipped with computation resource capable of hosting VNFs, and they communicate with each other via ISLs.
A distributed master-slave controller architecture is adopted in this paper to realize timely network status awareness, based on which VNF placement and routing planning are optimized.
The configuration details of the controllers are omitted and can be found in our previous work \cite{yuan2022softwaredefined}.
In the system model, various types of end-to-end services exist, including services that originate from ground users and are transmitted through satellites to ground stations, such as Internet of Things data collection services in remote areas, and services that originate from satellites and are delivered to ground stations or ground users, such as remote sensing services.

\subsection{Network Modeling with TEG}
The network configuration period for VNF placement and routing planning is divided into $T$ time slots, whose set is given by $\cal T$, and the length of each time slot $t \in {\cal T}$ is $\tau$.
Similar to \cite{jia2021vnfbasedservice,gao2022dynamicresource}, the network topology is assumed to be quasi-static within each time slot but varies across different time slots, due to the movement of LEO satellites.
Then, as shown in Fig. \ref{fig:des_TEGmodel}, the network is further characterized by a TEG ${\cal G} = \left( {{\cal N},{\cal L}} \right)$, where ${\cal N} = {\cal N}_s \cup {\cal N}_g \cup {\cal N}_u$ and ${\cal L} = {\cal L}_c \cup {\cal L}_s$ represent the set of network nodes and the set of links, respectively.
Specifically, ${\cal N}_s$, ${\cal N}_g$, and ${\cal N}_u$ denote the set of satellites, ground stations, and ground users, respectively, and the available computation resource at network node $n$ is represented as $C_n$.
${\cal L}_c$ and ${\cal L}_s$ are the set of communication links and stay links, respectively.
Specifically, the connection from node $n$ to node $m$ (represented as $nm$) is categorized as a direct communication link when $n \ne m$ and a stay link when $n = m$.
For cases where $n \ne m$, nodes $n$ and $m$ may be neighbors and capable of establishing an available direct communication link $nm$, or the significant distance between the two nodes may lead to direct communication link $nm$ being unavailable.
Stay links are used to characterize situations where services undergo processing or temporary storage for future forwarding at a network node within a single time slot.
Consequently, the parameters ${\varphi _{nm}^t}$, ${d _{nm}^t}$, and ${R_{nm}^t}$ denote the accessibility of link ${{nm}}$ during time slot $t$, the spatial distance between nodes $n$ and $m$, and the achievable transmission rate, respectively.
More precisely, ${\varphi _{nm}^t} = 1$ implies the availability of link ${{nm}}$ in time slot $t$, while a value of $0$ suggests inaccessibility.
Owing to the predictability of satellite movement trajectories, the value of ${\varphi _{nm}^t}$ for each communication link can be conveniently calculated in advance \cite{li2022processingwhiletransmittingcostminimized}.
In addition, it is evident that for each stay link $nn \in {\cal L}_s$, the values are established as ${\varphi _{nn}^t} = 1$, ${d _{nn}^t} = 0$, and ${R _{nn}^t} = \infty$.

\subsection{Service Model}
Assume that a set of end-to-end services occur simultaneously at the start of each configuration period, whose set is defined as ${\cal Q} = \{ q|q = 1,2,..., Q\}$.
Service $q \in {\cal Q}$ is described by a tuple $\{{s_q}, {e_q}, {\delta _q},{c_q},{{\cal I}_q}\}$, where ${{s_q}, {e_q}, \delta _q}$, and ${c_q}$ represent the source node, the destination node, the data size originated from the source node, and the computation resource requirement, respectively.
${{\cal I}_q}$ denotes a set of ordered VNFs which include two dummy VNFs and the set of VNFs required by service $q$.
These two dummy VNFs are introduced for the convenience of modeling. The first one is placed on the source node $s_q$, and the second one is placed on the destination node $e_q$.
$i \in {\cal I}_q$ indicates the $i$-th VNF of service $q$, and the number of VNFs in ${\cal I}_q$ is denoted by $I_q$.
To complete the service, the data flow of service $q$ needs to pass through all VNFs in ${\cal I}_q$ in order.
Moreover, we use ${{{\cal L}_q}}  = \{ l_{i} | i \in {{\cal I}_q} \backslash I_q \} $ to represent the set of virtual links experienced by service $q$.
Specifically, ${l_i} \in {{{\cal L}_q}} $ denotes the $i$-th virtual link of service $q$ which starts from $i$-th VNF to $(i+1)$-th VNF.

Define $x_n^{q,i}$ as a binary variable indicating whether the $i$-th VNF of service $q$ is placed on network node $n$, and we have
\begin{equation}
	x_n^{q,i}=
	\begin{cases}
		1, & \text{if } i \text{-th VNF of service } {q} \\
		   & \text{is deployed on node } n, \\
		0, & \text { otherwise. }
	\end{cases}
\end{equation}
Further, we use binary variable $y_{{nm}}^{q,i,t}$ to indicate whether the $i$-th virtual link $l_i \in {\cal L}_q$ passes through link ${nm} \in \cal L$ in time slot $t$, and we have
\begin{equation}
	y_{nm}^{q,i,t} =
	\begin{cases}
		1, & \text{if } i \text{-th virtual link} \text { passes through } {nm}, \\
		0, & \text {otherwise. }
	\end{cases}
\end{equation}

\subsection{Communication Model}

\subsubsection{ISL Model}
Inter-satellite communication operates in a free-space environment, which is mainly affected by free-space path loss and thermal noise \cite{zhou2019distributionallyrobust, leyva-mayorga2021interplaneintersatellite}.
Hence, following \cite{leyva-mayorga2021interplaneintersatellite, zhou2019distributionallyrobust}, the channel gain of ISL $nm \in {\cal L}_c$ in time slot $t$ with $n$ and $m$ representing two LEO satellites is given as
\begin{equation}
	g_{nm}^{t} =  {\frac{{c\sqrt {{G_{nm}}} }}{{4\pi {d_{nm}^t}f_{nm}\sqrt {{k_{{B}}}\zeta B_{nm}} }}},
\end{equation}
where $c$ is light speed, ${k_B}$ is Boltzmann constant, and $\zeta $ is thermal noise in Kelvin. $f_{nm}$, ${G_{nm}}$, and $B_{nm}$ are the center frequency, the fixed antenna gain, and the channel bandwidth of link $nm$, respectively.
Assume that the transceiver antennas of both satellites involved in an ISL can be aligned in the direction of maximum radiation \cite{leyva-mayorga2021interplaneintersatellite}.
Then, the available data rate of ISL $nm$ in time slot $t$ is expressed as
\begin{equation}
	R_{nm}^{t} = {B_{nm}}{\log _2}(1 + {p_{n}}|g_{nm}^t{|^2}),
\end{equation}
where $p_{n}$ is the transmission power of satellite $n$.

\subsubsection{STL Model}
The STL is modeled using Rician channel \cite{di2019ultradenseleoa}, and both dominant line of sight (LoS) path and weak scatter components propagated via different non-LoS (NLoS) paths are considered \cite{ippolito2008satellitecommunications,chen2021energyconstrainedcomputation}.
Hence, the channel gain from terrestrial node $n$ to satellite $m$ in time slot $t$ is given by \cite{chen2021energyconstrainedcomputation}
\begin{equation}
	\begin{aligned}
		g_{nm}^{t} =  \sqrt {\frac{\eta }{{1 + \eta }}} g_{nm}^{L,t} + \sqrt {\frac{1}{{1 + \eta }}{{\left( {{d_{nm}^t}} \right)}^{ - \rho '}}} g_{nm}^{NL,t},
	\end{aligned}
\end{equation}
where $g_{nm}^{L,t} = \sqrt {{{\left( {{d_{nm}^t}} \right)}^{ - \rho }}} {e^{ - j\frac{{2\pi f_{nm} }}{c }{d_{nm}^t}}}$ denotes the channel gain in LoS case while $g_{nm}^{NL,t} \sim {\cal N}(0,1)$ denotes the small-scale channel gain.
In addition, $\eta $ is the Rician factor, and $\rho $ and $\rho' $ are path loss exponents in the LoS case and NLoS case, respectively.
The achievable data rate $R_{{nm}}^t$ from terrestrial node $n$ to satellite $m$ and the achievable data rate $R_{{mn}}^t$ from satellite $m$ to terrestrial node $n$ are written as
\begin{equation}
	R_{{nm}}^t=B_{nm}\log _{2}(1+\frac{p_{n} G_{nm} \left|g_{{nm}}^{t}\right|^{2}}{\sigma_{0}^{2}B_{nm}}),
\end{equation}
\begin{equation}
	R_{mn}^t = B_{mn}{\log _2}( {1 + \frac{{{p_{m}}G_{mn}{{\left| {g_{mn}^{t}} \right|}^2}}}{{\sigma _{0}^2{B_{mn}}}}}),
\end{equation}
where ${p_{n}}$ and ${p_{m}}$ are the transmission power of terrestrial node $n$ and satellite $m$, respectively, and $\sigma _{0}$ is the Gaussian noise power.

\subsection{Latency Model}
\subsubsection{Communication Latency}
For each service, its communication latency, which is composed of transmission latency and propagation latency, is led by the data flow transmission over communication links.
Note that there is no communication latency when the service flow passes through stay links.
For service $q \in {\cal Q}$, the communication latency when the virtual link $l_i \in {\cal L}_q$ passes through communication link $nm \in {\cal L}_c$ is formulated as
\begin{equation}
	\label{eq:communitaionLatency}
	T_{nm}^{q,{l_i},t} = {\delta _q}\sum\limits_{q \in {\cal Q}} {\sum\limits_{i \in {{\cal I}_q}\backslash {I_q}} {y_{nm}^{q,i,t}} } /R_{nm}^t + \frac{{d_{nm}^t}}{v}.
\end{equation}
The first and second terms on the right side represent the transmission latency and propagation latency, respectively.
As indicated in the above formula, the bandwidth resource is assumed to be equally allocated among service flows that pass through the link in time slot $t$.

\subsubsection{Computation Latency}
For VNF $i \in {\cal I}_q$ of service $q$ placed on node $n$, the computation latency is
\begin{equation}
	T_{n}^{q,i} = { {\delta _q}\varepsilon}/{{{c_q}}},
\end{equation}
where $\varepsilon$ denotes the required computation resource to process one-bit data.
Since multiple VNFs of service $q$ may be placed on node $n$, the computation latency of service $q$ on node $n$ is calculated as
\begin{equation}
	T_{n}^q = \sum\limits_{i \in {{\cal I}_q}} {T_{n}^{q,i}} x_n^{q,i}.
\end{equation}

\subsection{Problem Formulation}
\subsubsection{Latency Constraints}
Assume data transmission must be finished within a single time slot when virtual link $l_i$ of service $q$ passes through communication link $nm$, which leads to the following constraint
\begin{equation}
	\label{eq:comLatency}
	T_{nm}^{q, l_i,t}y_{nm}^{q,i,t} \le \tau, \forall i \in {\cal I}_q \backslash I_q, q \in {\cal Q}, nm \in {\cal L}_c, t \in {\cal T}.
\end{equation}
On the other hand, the cumulative stay time of service $q$ on node $n$ should be longer than the computation latency of service $q$ on node $n$, which results in
\begin{equation}
	\label{eq:CompLatencyCons}
	T_n^q \le \sum\limits_{i \in {{{\mathcal I}_q}\backslash I_q} } {\sum\limits_{t \in {\mathcal T}} {y_{nn}^{q,i,t}\tau } }, \forall q \in {\mathcal Q}, n \in {\mathcal N},
\end{equation}
where the right-hand term is the cumulative stay time of service $q$ on node $n$, which is decided by the number of time slots that the service flow stays on link $nn$.

\subsubsection{VNF Placement Constraints}
For service $q$, each required VNF is assumed to be deployed on only one network node, i.e.,
\begin{equation}
	\label{eq:vnfEmbedCons}
	\sum \limits_{n \in {\cal N}} x_n^{q,i} = 1, \forall i \in {{\cal I}_q},q \in {\cal Q}.
\end{equation}
Note that the first VNF and the last VNF of service $q$ are placed on source node $s_q$ and destination node $e_q$, respectively, and that is
\begin{equation}
	\label{eq:vnfEmbedCons2}
	x_{{s_q}}^{q, 1} = 1, x_{{e_q}}^{q, I_q} = 1, \forall q \in {\cal Q}.
\end{equation}

\subsubsection{Flow Constraints}
The data flow of service $q$ goes into the network at the first time slot of the configuration period from source node $s_q$, which indicates the following constraint
\begin{equation}
	\label{eq:ingressLinkCons}
	\sum\limits_{n \in {\cal N}} {y_{{s_q}n}^{q,1,1}}  = 1,\forall q \in {\cal Q}.
\end{equation}
In addition, the data flow of each service should traverse all the nodes holding the required VNFs, which means
\begin{equation}
	\label{eq:routingVNF}
	x_n^{q,i} \le {\sum\limits_{t \in {\mathcal{T}}} {y_{nn}^{q,i,t}} } ,\forall q \in {\cal Q}, n \in {\cal N},i \in {{\cal I}_q}\backslash \{1,I_q\}.
\end{equation}
Meanwhile, within each time slot, the data flow of each service can only traverse either a communication link or a stay link.
When the service path traverses a communication link, it results in $\sum_{nm \in \mathcal{L}_c}{\sum_{i\in \mathcal{I}_q\backslash I_q}{y_{nm}^{q,i,t}}} = 1$.
On the other hand, when the service path traverses a stay link, the potential arises for deploying multiple consecutive VNFs on a single node and accomplishing their processing within a single time slot, which means $\sum_{i\in \mathcal{I} q\backslash I_q}{y_{nn}^{q,i,t}}\ge 1$.
By combining these two cases, the following constraint can be derived as
\begin{equation}
	\label{eq:comOneTimeSlotCons}
	\begin{aligned}
		\sum_{nm\in \mathcal{L} _c}{\sum_{i\in \mathcal{I} _q\backslash I_q}{y_{nm}^{q,i,t}}}+\sum_{nn\in \mathcal{L} _s}{\mathbb{I}}(\sum_{i\in \mathcal{I} _q\backslash I_q}{y_{nn}^{q,i,t}}\ge 1)\le 1, \\
		\forall q\in \mathcal{Q} ,t\in \mathcal{T},
	\end{aligned}
\end{equation}
where $\mathbb{I}( \cdot )$ is the indicator function.

To ensure that traffic flowing in and out of each network node $n$ is equal and consecutively passes the required VNFs in the correct order, flow conservation constraints are essential \cite{wang2020sfcbasedservice, jia2021vnfbasedservice}.
However, the number of time slots needed to deliver service data from the source node to the destination node is uncertain, which poses a challenge in formulating flow conservation constraints.
To overcome this issue, we assume that service $q$ is completed in time slot ${t_q}$, with $1 \le {t_q} \le T$ \footnote{In this paper, we assume that all services can be completed during the configuration period. The network controller can split services spanning multiple periods prior to planning.}.
Then, we formulate the dynamic flow conservation constraints as follows.

\begin{itemize}
	\begin{subequations}
		\label{eq:flowCons}
		\item
		\textbf{Case 1: ($t<{t_q}, i < { I_q -1}$)}

		In this case, the total traffic of $i$-th virtual link goes into node $n$ in time slot $t$ is given by $\sum_{m \in {\cal N}} {{y_{mn}^{q,{i},t}}}$. As for the total traffic that goes out from node $n$, there are three situations.
		First, the $i$-th virtual link ends in time slot $t$ at node $n$ and the $i+1$-th virtual link starts at node $n$ in time slot $t+1$.
		The total traffic of service $q$ going out of node $n$ is given by $\sum_{m \in {\cal N}} {( { y_{nm}^{q,{{i + 1}},t + 1} })}$.
		Second, the $i$-th virtual link ends in time slot $t$ at node $n$ and the $i+1$-th virtual link starts at node $n$ also in time slot $t$. The total traffic of service $q$ that flows out is expressed as $\sum_{m \in {\cal N}} {( { y_{nm}^{q,{{i+1}},t} })}$.
		Third, the $i$-th virtual link does not end in time slot $t$ at node $n$ and the total traffic of service $q$ going out of node $n$ is written as $\sum_{m \in {\cal N}} {( { y_{nm}^{q,{{i}},t + 1} })}$.
		Since only one of these three situations happens, the flow conservation constraint in Case 1 is derived as

		\begin{small}
			\begin{equation}
				\label{eq:flowCons1}
				\begin{aligned}
					\sum\limits_{m \in {\cal N}} {{y_{mn}^{q,{i},t}}}  = \sum\limits_{m \in {\cal N}} {( {y_{nm}^{q,{{i + 1}},t} + y_{nm}^{q,{{i + 1}},t + 1} + y_{nm}^{q,{i},t + 1}})} , \\
					\forall t < {t_q}, i < { I_q -1}, q \in {\cal Q}, n \in {\cal N}.
				\end{aligned}
			\end{equation}
		\end{small}

		\item
		\textbf{Case 2: ($t<{t_q}, i = {{ I_q -1}}$)}

		Since the current virtual link is the last one and the service has not been finished yet, the corresponding flow conservation constraint is given as

		\begin{small}
			\begin{equation}
				\label{eq:flowCons2}
				\begin{aligned}
					\sum\limits_{m \in {\cal N}} {{y_{mn}^{q,i,t}}} = \sum\limits_{m \in {\cal N}} {{y_{nm}^{q,i,t + 1}}} , i = I_q -1, \\
					\forall t < {t_q}, q \in {\cal Q}, n \in {\cal N}.
				\end{aligned}
			\end{equation}
		\end{small}

		\item
		\textbf{Case 3: (${t_q} < t$)}

		There is no traffic of service $q$ in the network after service $q$ is finished, which means

		\begin{small}
			\begin{equation}
				\label{eq:flowCons3}
				\sum\limits_{m \in {\cal N}} {{y_{mn}^{q,i,t}}}  = 0, \forall {t_q} < t, i \le I_q -1, q \in {\cal Q}, n \in {\cal N}.
			\end{equation}
		\end{small}

		\item
		\textbf{Case 4: ($t = {t_q},i < I_q -1$)}

		Since service $q$ is finished in time slot $t$ and the current virtual link is not the last one, flow conservation is kept by the following constraint

		\begin{small}
			\begin{equation}
				\label{eq:flowCons4}
				\begin{aligned}
					\sum\limits_{m \in {\cal N}} {{y_{mn}^{q,i,t}}}  = \sum\limits_{m \in {\cal N}} {{y_{nm}^{q,i + 1,t}}} , t = {t_q}, \\
					\forall i <  I_q -1, q \in {\cal Q}, n \in {\cal N}.
				\end{aligned}
			\end{equation}
		\end{small}

		\item
		\textbf{Case 5: ($t = {t_q},i = I_q -1$)}

		In this case, the last virtual link of service $q$ ends at the destination node $e_q$. The corresponding flow conservation constraint is formulated as

		\begin{small}
			\begin{equation}
				\label{eq:flowCons5}
				\begin{aligned}
					 & \sum\limits_{m \in {\cal N}} {{y_{mn}^{q,i,t}}}  = 1,t = {t_q},i = I_q -1 ,n = {e_q}, \forall q \in {\cal Q},   \\
					 & \sum\limits_{m \in {\cal N}} {{y_{mn}^{q,i,t}}}  = 0,t = {t_q},i = I_q -1 ,n \ne {e_q}, \forall q \in {\cal Q}.
				\end{aligned}
			\end{equation}
		\end{small}

	\end{subequations}
\end{itemize}

Finally, a virtual link of a service flow can pass through link $nm$ in time slot $t$ if and only if link $nm$ is available in time slot $t$, and that is
\begin{equation}
	\label{eq:linkstatusCons}
	\mathbb{I}(\sum\limits_{q \in {\cal Q}} {\sum\limits_{i \in {{{\cal I}_q}\backslash I_q} } {y_{nm}^{q,i,t}} }  \ge 1 ) \le \varphi _{nm}^t, \forall nm \in {\cal L},t \in {\cal T}.
\end{equation}

\subsubsection{Resource Constraints}
For each network node, the computation resource consumption in each time slot incurred by VNF placement and service data processing cannot exceed total resource limitation $C_n$, which means
\begin{equation}
	\label{eq:computationCons}
	\sum\limits_{q \in {\cal Q}} {\sum\limits_{{i} \in {{\cal I}_q}\backslash {I_q} } {x_n^{q,i}( {\beta _i^q + y_{nn}^{q,i,t}{c_q}})} } \le {C_n},\forall n \in {\cal N},t \in {\cal T},
\end{equation}
where $\beta _i^q$ denotes the required computation resource for holding $i$-th VNF of service $q$.


\subsubsection{Formulated Problem}
The end-to-end latency of service $q$ is calculated by
\begin{equation}
	\label{eq:e2elatencyCal}
	T_{e2e}^q = \tau {\sum\limits_{t \in {\cal T}}  {\mathbb{I}( {\sum\limits_{i \in {{{\cal I}_q}\backslash I_q}} {\sum\limits_{nm \in {\cal L}} {y_{nm}^{q,i,t}} } } \ge 1)}}, \forall q \in {\cal Q}.
\end{equation}
The objective of this work is to minimize the average end-to-end latency of all services in the SDSTN by joint VNF placement and routing planning optimization, and the concerned problem is formulated as
\begin{equation}
	\label{eq:primalProblem}
	\begin{aligned}
		(\rm{P1}):  \quad
		\mathop {\min }\limits_{\boldsymbol{\rm{x,y}}} \quad
		 & \frac{1}{Q}\sum\limits_{q \in {\cal Q}} T_{e2e}^q               \\
		\quad \rm{s.t.} \quad
		 & (\text{\ref{eq:comLatency}})-(\text{\ref{eq:computationCons}}).
	\end{aligned}
\end{equation}
It is observed that P1 is an INLP due to the nonlinear constraints and integer variables.
In \cite{jia2021vnfbasedservice,gao2022virtualnetwork,li2022costawaredynamic}, it has been proven that joint VNF placement and routing planning is NP-hard in general, which is nontrivial to solve especially in a time-slotted SDSTN.

\section{BDBC Algorithm for Optimal Solution}
\label{sec:bcbd4exactSolution}

Benders decomposition (BD) is an advanced decomposition method that exploits the structure of integer programming problems and is capable of reducing the size of integer subproblems by fixing specific variables and generating cuts to tighten the solution space.
In this section, we propose a Benders decomposition based branch-and-cut (BDBC) algorithm to address the above problem.
As a prerequisite for using BD, we first transform the primal problem into an ILP.
Subsequently, the ILP is decomposed into a master problem and a subproblem.
To handle the integer variables in the subproblem, a branch-and-bound method combined with BD is designed.

\subsection{Transformation from INLP to ILP}

\subsubsection{The Linearization of Constraint (\ref{eq:comLatency})}
Constraint (\ref{eq:comLatency}) can be described as
\begin{equation}
	\label{eq:trans4comLatency}
	\begin{aligned}
		 & \left\{{
				\begin{array}{*{20}{l}}
					\text{IF } y_{nm}^{q,i,t} = 1, \text{THEN } T_{nm}^{q, l_i,t} \le \tau, \\
					\text{IF } y_{nm}^{q,i,t} = 0, \text{THEN } 0 \le \tau.
				\end{array}
		}\right.                                                                                                     \\
		 & \hspace{0.5in}\forall i \in {\cal I}_q \backslash I_q, q \in {\cal Q}, nm \in {\cal L}_c, t \in {\cal T}.
	\end{aligned}
\end{equation}
The first IF-THEN expression can be equivalently transformed into the following constraint
\begin{equation}
	\label{eq:comLatencyLinearCons}
	\begin{aligned}
		 & T_{nm}^{q, l_i,t} \le \tau + \tau T(1 - y_{nm}^{q,i,t}),                                                  \\
		 & \hspace{0.5in}\forall i \in {\cal I}_q \backslash I_q, q \in {\cal Q}, nm \in {\cal L}_c, t \in {\cal T}.
	\end{aligned}
\end{equation}
For the second IF-THEN expression, it has no impact on solution space and hence can be omitted.

\subsubsection{The Linearization of Constraint (\ref{eq:comOneTimeSlotCons})}
In order to linearize constraint (\ref{eq:comOneTimeSlotCons}), auxiliary variable $f_{nn}^{q,t}$ is introduced, which is defined as
\begin{equation}
	\label{eq:axV4routingwoVNF}
	f_{nn}^{q,t} = \mathbb{I}( {\sum\limits_{i \in {{\cal I}_q}\backslash I_q} {y_{nn}^{q,i,t}}  \ge 1}), \forall nn \in {\cal L}_s,q \in {\cal Q},t \in {\cal T}.
\end{equation}
By substituting (\ref{eq:axV4routingwoVNF}) into constraint (\ref{eq:comOneTimeSlotCons}), we obtain the linearized constraint
\begin{equation}
	\label{eq:routingwoCons}
	\begin{aligned}
		\sum\limits_{nm \in {{\cal L}_c}} {\sum\limits_{i \in {{\cal I}_q}\backslash {I_q}} {y_{nm}^{q,i,t}} }  + \sum\limits_{nn \in {{\cal L}_s}} f_{nn}^{q,t}  \le 1, \forall q \in {\cal Q},t \in {\cal T}.
	\end{aligned}
\end{equation}
Furthermore, formula (\ref{eq:axV4routingwoVNF}) can be transformed into a linear form by Lemma \ref{lm:lemma1}.
\begin{lemma}
	\label{lm:lemma1}
	Formula (\ref{eq:axV4routingwoVNF}) with indicator function can be re-written into a linear form as
	\begin{equation}
		\label{eq:linearizeF}
		\begin{aligned}
			\left\{{
					\begin{array}{*{20}{l}}
						{f_{nn}^{q,t} \in \{0,1\}}
						\\
						{\sum\limits_{i \in {{\cal I}_q}\backslash I_q} {y_{nn}^{q,i,t}}  \le {{L}_q}f_{nn}^{q,t}}
						\\
						{f_{nn}^{q,t} \le \sum\limits_{i \in {{\cal I}_q}\backslash I_q} {y_{nn}^{q,i,t}} }
					\end{array}
				} \right.,
			\forall nn \in {\cal L}_s,q \in {\cal Q},t \in {\cal T}.
		\end{aligned}
	\end{equation}
\end{lemma}

\begin{proof}
	The first sub-constraint in (\ref{eq:linearizeF}) indicates that $f_{nn}^{q,t}$ is either 0 or 1, which is consistent with (\ref{eq:axV4routingwoVNF}). When $f_{nn}^{q,t}=0$, (\ref{eq:linearizeF}) results in $\sum_{i \in {{\cal I}_q}\backslash I_q} {y_{nn}^{q,i,t}}=0$, which is the same as (\ref{eq:axV4routingwoVNF}).
	When $f_{nn}^{q,t}=1$, compared to (\ref{eq:axV4routingwoVNF}), although (\ref{eq:linearizeF}) additionally results in
	$\sum_{i \in {{\cal I}_q}\backslash I_q} {y_{nn}^{q,i,t}} \le {L}_q$ except $1 \le \sum_{i \in {{\cal I}_q}\backslash I_q} {y_{nn}^{q,i,t}}$, $\sum_{i \in {{\cal I}_q}\backslash I_q} {y_{nn}^{q,i,t}} \le {L}_q$ has no impact on primal solution space, since ${L}_q$ is the maximum number of virtual links of service $q$.
\end{proof}

\subsubsection{The Linearization of Constraint (\ref{eq:flowCons})}
First, auxiliary variable $h_q^t$ is introduced, which is defined as
\begin{equation}
	\label{eq:threePartLinear}
	h_q^t =
	\begin{cases}
		1, & \text{service } q \text { is completed in time slot } t, \\
		0, & \text{otherwise. }
	\end{cases}
\end{equation}
and satisfies
\begin{equation}
	\label{eq:auxTimeDivide}
	\sum\limits_{t \in {\cal T}} {h_q^t}  = 1, \forall q \in {\cal Q}.
\end{equation}
Second, another auxiliary binary variable $k_q^t$ is introduced and given by
\begin{equation}
	\label{eq:middleAuxiliary}
	k_q^t = \mathbb{I}(\sum\limits_{t' \le t} {h_q^{t'}} = 1), \forall t \in {\cal T}, q \in {\cal Q}.
\end{equation}
It is obvious that $k_q^t = 0$ when service $q$ has not been finished in time slot $t$.
Similar to constraint (\ref{eq:axV4routingwoVNF}), formula (\ref{eq:middleAuxiliary}) can be linearized as
\begin{equation}
	\label{eq:linearK}
	\left\{ {\begin{array}{*{20}{l}}
				{k_q^t \in \{0,1\}}                              \\
				{\sum\limits_{t' \le t} {h_q^{t'}}  \le T k_q^t} \\
				{k_q^t \le \sum\limits_{t' \le t} {h_q^{t'}} }
			\end{array}} \right.,\forall q \in {\cal Q},t \in {\cal T}.
\end{equation}
Consequently, the three categories of time slots for service $q$ in (\ref{eq:flowCons}) during the configuration period can be represented as given below.
\begin{subequations}
	\begin{align}
		\begin{split}
			&t<t_q\leftrightarrow 1-k_{q}^{t}=1,
		\end{split}
		\\
		\begin{split}
			&t>t_q\leftrightarrow k_{q}^{t}-h_{q}^{t}=1,
		\end{split}
		\\
		\begin{split}
			&t=t_q\leftrightarrow h_{q}^{t}=1.
		\end{split}
	\end{align}
\end{subequations}

Then, the big-$M$ method is utilized to linearize constraint (\ref{eq:flowCons}), which is given by

\begin{small}
	\begin{subequations}
		\label{eq:flowConsLinear}
		\begin{align}
			\begin{split}
				&\begin{aligned}
					0\le \sum_{m\in \mathcal{N}}{\left( y_{mn}^{q,i,t}-y_{nm}^{q,i+1,t}-y_{nm}^{q,i,t+1}-y_{nm}^{q,i+1,t+1} \right)}\le k_{q}^{t}, \\
					\forall t\in \mathcal{T} \backslash T,i\in \mathcal{I} _q\backslash \{I_q-1,I_q\},
				\end{aligned}
			\end{split}
			\\
			\begin{split}
				0\le \sum_{m\in \mathcal{N}}{\left( y_{mn}^{q,i,t}-y_{nm}^{q,i,t+1} \right)}\le k_{q}^{t},\forall t\in \mathcal{T} \backslash T,i=I_q-1,
			\end{split}
			\\
			\begin{split}
				0\le \sum_{m\in \mathcal{N}}{y_{mn}^{q,i,t}}\le (1-(k_{q}^{t}-h_{q}^{t})),\forall t\in \mathcal{T} ,i\in \mathcal{I} _q\backslash I_q,
			\end{split}
			\\
			\begin{split}
				\begin{aligned}
					 & 0\le \sum_{m\in \mathcal{N}}{\left( y_{mn}^{q,i,t}-y_{nm}^{q,i+1,t} \right)}\le (1-h_{q}^{t}), \\
					 & \hspace{1.5in}\forall t\in \mathcal{T} ,i\in \mathcal{I} _q\backslash \{I_q-1,I_q\},
				\end{aligned}
			\end{split}
			\\
			\begin{split}
				\begin{aligned}
					 & 0\le \sum_{m\in \mathcal{N}}{\left( y_{mn}^{q,i,t}-1 \right)}\le (1-h_{q}^{t}), \\
					 & \hspace{1.5in}\forall t\in \mathcal{T} ,i=I_q-1,n=e_q,
				\end{aligned}
			\end{split}
			\\
			\begin{split}
				0\le \sum_{m\in \mathcal{N}}{y_{mn}^{q,i,t}}\le (1-h_{q}^{t}),\forall t\in \mathcal{T} ,i=I_q-1,n\ne e_q.
			\end{split}
		\end{align}
	\end{subequations}
\end{small}

\subsubsection{The Linearization of Constraint (\ref{eq:linkstatusCons})}
Constraint (\ref{eq:linkstatusCons}) can be transformed into a linear form by Lemma \ref{lm:lemma2}, and the equivalence between (\ref{eq:linkstatusCons}) and (\ref{eq:linearizeVar}) can be proved in a similar way of Lemma \ref{lm:lemma1}.
\begin{lemma}
	\label{lm:lemma2}
	Constraint (\ref{eq:linkstatusCons}) can be linearized as
	\begin{equation}
		\label{eq:linearizeVar}
		\sum\limits_{q \in {\cal Q}} {\sum\limits_{i \in {{\cal I}_q}\backslash I_q} {y_{nm}^{q,i,t}} }  \le \sum\limits_{q \in {\cal Q}} {{{ L}_q}} \varphi _{nm}^t,\forall nm \in {\cal L},t \in {\cal T}.
	\end{equation}
\end{lemma}

\subsubsection{The Linearization of Constraint (\ref{eq:computationCons})}
To linearize constraint (\ref{eq:computationCons}), variable $b_n^{q,i,t} = x_n^{q,i}y_{nn}^{q,i,t}$ is introduced.
Then, constraint (\ref{eq:computationCons}) is transformed into
\begin{equation}
	\label{eq:computationLinearizedCons}
	\sum\limits_{q \in {\cal Q}} {{\sum _{{i} \in {{\cal I}_q}\backslash {I_q}}}(x_n^{q,i}\beta _i^q + b_n^{q,i,t}{c_q})}  \le {C_n},\forall n \in {\cal N},t \in {\cal T}.
\end{equation}
Meanwhile, $b_n^{q,i,t}$, $x_n^{q,i}$, and $y_{nn}^{q,i,t}$ should satisfy
\begin{equation}
	\label{eq:bntCons}
	\begin{aligned}
		b_n^{q,i,t} \le x_n^{q,i}, b_n^{q,i,t} \le y_{nn}^{q,i,t}, b_n^{q,i,t} \ge x_n^{q,i} + y_{nn}^{q,i,t} -1, \\
		b_n^{q,i,t} \in \{0,1\}, \forall q \in {\cal Q}, i \in {{\cal I}_q}\backslash {I_q}, n \in {\cal N}, t \in {\cal T}.
	\end{aligned}
\end{equation}

\subsubsection{The Linearization of Objective (\ref{eq:e2elatencyCal})}
Auxiliary variable $a_q^t$ is introduced to linearize (\ref{eq:e2elatencyCal}), which is defined as
\begin{equation}
	\label{eq:objAuxiValue}
	a_q^t = {{\mathbb{I}( {\sum\limits_{i \in {{{\cal I}_q}\backslash I_q}} {\sum\limits_{nm \in {\cal L}} {y_{nm}^{q,i,t}} } } \ge 1)}}, \forall q \in {\cal Q}.
\end{equation}
Subsequently, the indicator function in (\ref{eq:objAuxiValue}) can be transformed into a linear form by referencing Lemma \ref{lm:lemma1}, which is given as
\begin{equation}
	\label{eq:linearizedA}
	\begin{aligned}
		\left\{ {\begin{array}{*{20}{l}}
					         {a_q^t \in \{0,1\}}                                                                                                   \\
					         \sum\limits_{i \in {{\cal I}_q}\backslash I_q} {\sum\limits_{nm \in {\cal L}} {y_{nm}^{q,i,t}} }  \le {{L_q}} L a_q^t \\
					         {a_q^t \le \sum\limits_{i \in {{\cal I}_q}\backslash I_q} {\sum\limits_{nm \in {\cal L}} {y_{nm}^{q,i,t}} } }
				         \end{array}} \right.,\forall q \in {\cal Q}, t \in {\cal T}.
	\end{aligned}
\end{equation}

\subsubsection{ILP Formulation of Primal Problem}

As a result, problem P1 can be reformulated as an ILP as follows.
\begin{equation}
	\begin{aligned}
		({\rm{P2}}):  \quad
		\mathop {\min } \limits_{{\cal {Z}}} \quad
		 & \frac{1}{Q}\sum\limits_{q \in {\cal Q}} T_{e2e}^q                                                                                                         \\
		{\rm{s.t.}} \quad
		 & (\ref{eq:CompLatencyCons})-(\ref{eq:routingVNF}), (\ref{eq:comLatencyLinearCons}), (\ref{eq:routingwoCons})-(\ref{eq:auxTimeDivide}), (\ref{eq:linearK}), \\
		 & (\ref{eq:flowConsLinear})-(\ref{eq:bntCons}), (\ref{eq:linearizedA}),
	\end{aligned}
\end{equation}
where $\cal {Z} = \{\boldsymbol{\rm{x,y,f,h,k,a,b}}\}$.
Note that the computational complexity of an exhaustive search for the optimal solution of P2 is determined by the number of integer variables and the number of constraints.
Actually, we can find that the number of binary integer variables in P2 is

\begin{small}
	\begin{equation}
		\begin{aligned}
			V
			 & = N\sum\limits_{q \in {\cal Q}} {{I_q}}  + (L + N)T\sum\limits_{q \in {\cal Q}} {({I_q} - 1)}  + (3 + L)QT \\
			 & \cong LT\sum\limits_{q \in {\cal Q}} {({I_q} - 1)},
		\end{aligned}
	\end{equation}
\end{small}
and the number of constraints is
\begin{equation}
	\begin{aligned}
		C =
		      & QN + \sum\limits_{q \in {\cal Q}} {{I_q}}  + 4Q + QT(3N + 2L + 3) + LT                             \\
		      & + N\sum\limits_{q \in {\cal Q}} {({I_q} - 2)}  + NT + 4NT\sum\limits_{q \in {\cal Q}} ({I_q}  - 1) \\
		\cong & TL(1 + 2Q) ,
	\end{aligned}
\end{equation}
where $\cong$ denotes the approximation operation.
It is obvious that the exhaustive search with computational complexity $\mathcal{O}(V \cdot C \cdot 2^{V})$ \cite{jia2021vnfbasedservice} is unpractical for P2.

\subsection{Details of BDBC Algorithm}
BD is an efficient methodology to solve optimization problems with mixed variables, such as mixed-integer problems, by decomposing the primal problem into an ILP master problem and an LP subproblem \cite{benders2005partitioningprocedures}.
By solving the subproblem in each iteration of BD, information about the solution to the master problem can be obtained, which can be used to further tighten the solution space of the master problem.
Specifically, if the subproblem is infeasible, Benders cuts are generated and added to the master problem in the next iteration until both the subproblem and master problem converge to the optimal solution.
Due to the linear objective function of P2, which only depends on $a_q^t$, we can intuitively decompose P2 using BD into a master problem for routing planning and a subproblem for VNF placement.
However, the original BD algorithm requires the variables in subproblems to be continuous, which is not applicable in our case.

To overcome this issue, we refer to works that handle the integer variables of subproblem in BD \cite{sherali2002modificationbenders, sen2006decompositionbranchandcut, fakhri2017bendersdecomposition}, and then present a BDBC algorithm.
Specifically, a branch-and-bound tree is built to deal with the integer variables in the subproblem, and the BD algorithm is used to iteratively solve the master problem and relaxed subproblem for each node on the tree.
Since the objective of P2 is independent of the variables of the subproblem, it implies that the subproblem aims to find a feasible VNF placement on a given routing path.
Therefore, we can take the average VNF placement load at network nodes as the objective to generate the feasibility cuts, which does not affect the optimality of the solution to P2.
To save space, the details of the BD algorithm are omitted, and readers can refer to \cite{wolsey2021integerprogrammin}.
The master problem (MP) and relaxed subproblem (RSP) after decomposition are given as follows.
\begin{equation}
	\begin{aligned}
		(\rm{MP}):  \quad
		\mathop {\min } \limits_{{\cal {Z'}}} \quad
		 & \frac{1}{Q}\sum\limits_{q \in {\cal Q}} T_{e2e}^q                                                                                                                                                               \\
		\rm{s.t.} \quad
		 & (\ref{eq:ingressLinkCons}), (\ref{eq:comLatencyLinearCons}), (\ref{eq:routingwoCons})-(\ref{eq:auxTimeDivide}), (\ref{eq:linearK}), \\
		 & (\ref{eq:flowConsLinear}), (\ref{eq:linearizeVar}), (\ref{eq:linearizedA}),
	\end{aligned}
\end{equation}
where $\cal {Z'} = \{\boldsymbol{\rm{y,a,f,h,k}}\}$.
\begin{equation}
	\begin{aligned}
		(\rm{RSP}):  
		\mathop {\min } \limits_{{\boldsymbol{\rm{x}}}} \quad
		 & \sum\limits_{n \in {\cal N}} {\frac{1}{{{C_n}}}(\sum\limits_{q \in {\cal Q}} {\sum\limits_{i \in {I_q}} {x_n^{q,i}\beta _q^i} } )} \\
		\rm{s.t.} \quad
		 & (\ref{eq:CompLatencyCons}), (\ref{eq:vnfEmbedCons}), (\ref{eq:routingVNF}), (\ref{eq:computationCons}),                            \\
		 & 0 \le x^{q,i}_n \le 1, \forall q \in {\cal Q}, i \in {\cal I}_q, n \in {\cal N}.
	\end{aligned}
\end{equation}
Since $\boldsymbol{\rm{x}}$ obtained from solving RSP may not be an integer solution, which is not feasible for P2, a search tree is built and maintained for branch-and-bound algorithm based on $\boldsymbol{\rm{x}}$.
According to \cite{fakhri2017bendersdecomposition}, the Benders cuts generated at a node of the search tree are also effective for child nodes grown from this node.
Hence, the child node in the tree is able to inherit the Benders cuts of all the direct ancestor nodes to further tighten the solution space of the master problem, which is the basis of the branch-and-cut in the BDBC algorithm.

\begin{algorithm}[thbp]
	\begin{footnotesize}
		\DontPrintSemicolon
		\caption{Benders decomposition based branch-and-cut (BDBC) algorithm}
		\label{alg:BDBC}
		Create an empty tree node list $\mathbb{T}$, and initiate upper bound $\mathbf{\overline{B}} = \infty$ and lower bound $\mathbf{\underline{B}} = -\infty$\;
		Use BD algorithm to solve MP and RSP to obtain the optimal value $\mathbf{{V}}$, and pack Benders cuts generated during BD iterations together with $\mathbf{{V}}$ and the solution $\boldsymbol{\rm{x}}$ as a node, which is then added to $\mathbb{T}$\;
		\While{$\mathbb{T} \ne \emptyset$}{
		Select a node with Benders cuts, optimal value $\mathbf{{V}}$ and solution $\boldsymbol{\rm{x}}$ from $\mathbb{T}$\;
		Set $\mathbf{\underline{B}} = \mathbf{{V}}$\;
		\If{$\mathbf{\overline{B}} - \mathbf{\underline{B}} > e$}{
		\If{$\boldsymbol{\rm{x}}$ is integer}{
			Update the upper bound $\mathbf{\overline{B}} = \mathbf{{V}}$\;
			Remove the selected node from $\mathbb{T}$\;
		}
		\Else{
		{$/*$ \textit{Branch} $*/$}\;
		Branch on a non-integer variable $x'_{ik}$ in $\boldsymbol{\rm{x}}$ according to $\{{x_{ik}} \le \lfloor {x'_{ik}} \rfloor, {x_{ik}} \ge \lfloor {x'_{ik}} +1 \rfloor \}$\;
		/* \textit{Parallel Computing} */ \;
		\For{$\textbf{C} \in \{{x_{ik}} \le \lfloor {x'_{ik}} \rfloor, {x_{ik}} \ge \lfloor {x'_{ik}} +1 \rfloor \}$}{
		Inherit all the Benders cuts of the selected node to construct AMP, and add the constraint $\textbf{C}$ to construct ARSP\;
		Use the BD algorithm to solve the AMP and the ARSP to obtain the optimal value $\mathbf{{V}}'$\;
		Pack the Benders cuts generated during the BD iterations and the Benders cuts inherited from the ancestor nodes together with $\mathbf{{V}}'$ and the solution $\boldsymbol{\rm{x}}'$ as a node, which is then added to $\mathbb{T}$\;
		}
		Remove the selected node from $\mathbb{T}$\;
		}
		}
		\Else{
			Remove the selected node from $\mathbb{T}$\;
		}
		}
		$\mathbf{{B}^*} = \mathbf{\overline{B}}$\;
	\end{footnotesize}
\end{algorithm}

Considering branching for the non-integer ${x'_{ik}}$ of the solution of node $k$ in the search tree, the left branch is set as ${x_{ik}} \ge \lfloor {x'_{ik}} +1 \rfloor$ which indicates ${x_{ik}}=1$, and the right branch is set as ${x_{ik}} \le \lfloor {x'_{ik}} \rfloor$ which indicates ${x_{ik}}=0$, where $ik$ is the index of the branching value in the solution $\boldsymbol{\rm{x}}$ of node $k$.
Then, the augmented relaxed subproblem of node $k$ in the search tree, ARSP($k$), is given as
\begin{equation}
	\begin{aligned}
		(\text{ARSP}{(k)}):
		\mathop {\min } \limits_{{\bf{x}}}
		 & \sum\limits_{n \in {\cal N}} {\frac{1}{{{C_n}}}(\sum\limits_{q \in {\cal Q}} {\sum\limits_{i \in {{\cal I}_q}} {x_n^{q,i}\beta _q^i} } )} \\
		\rm{s.t.} \quad
		 & (\ref{eq:CompLatencyCons}), (\ref{eq:vnfEmbedCons}), (\ref{eq:routingVNF}), (\ref{eq:computationCons}),                                   \\
		 & x_{ik} \ge \lfloor {x'_{ik}} +1\rfloor, \forall k \in {\cal S}_k^1,                                                                       \\
		 & x_{ik} \le \lfloor {x'_{ik}} \rfloor, \forall k \in {\cal S}_k^2,                                                                         \\
		 & 0 \le x^{q,i}_n \le 1, \forall q \in {\cal Q}, i \in {\cal I}_q, n \in {\cal N},
	\end{aligned}
\end{equation}
where $ {\cal S}_k^1$ is the set of the direct ancestors of node $k$ in the search tree on the left branch, and ${\cal S}_k^2$ is the set of the direct ancestors of node $k$ in the search tree on the right branch.
By solving the dual problem of ARSP($k$), we can obtain the dual variables $\boldsymbol{\rm{\lambda}}$, $\boldsymbol{\rm{\alpha}}$, $\boldsymbol{\rm{\gamma}}$, $\boldsymbol{\rm{\theta}}$, $\boldsymbol{\rm{u}}$, and $\boldsymbol{\rm{v}}$, which correspond to (\ref{eq:CompLatencyCons}), (\ref{eq:vnfEmbedCons}), (\ref{eq:routingVNF}), (\ref{eq:computationCons}), the left and right branching operations, respectively.

We assume that the extreme points and the extreme rays of the dual problem of ARSP($k$) are ${\cal E}_k^p$ and ${\cal E}_k^r$, respectively.
Since the primal subproblem is a feasibility problem, only the feasibility cuts are considered.
Then, the following Benders cuts are generated by node $k$ for the master problem after the BD algorithm for node $k$ is completed.
\begin{equation}
	F_\text{y}(\boldsymbol{\rm{\lambda,\alpha,\gamma,\theta,u,v}}) \le 0, \forall (\boldsymbol{\rm{\lambda,\alpha,\gamma,\theta,u,v}}) \in {\cal E}_k^r.
\end{equation}
Similarly, the Benders cuts generated by node $k'$, which is one direct ancestor of node $k$, can be inherited by node $k$ with the following forms,
\begin{equation}
	F_\text{y}(\boldsymbol{\rm{\lambda,\alpha,\gamma,\theta,u',v'}}) \le 0, \forall (\boldsymbol{\rm{\lambda,\alpha,\gamma,\theta,u',v'}}) \in {\hat{{\cal E}}}_{k'}^r,
\end{equation}
where $\hat{{\cal E}}_{k'}^r$ is the modified version of the extreme rays of the dual problem of ARSP($k'$), which is denoted as ${{\cal E}}_{k'}^r$.
Specifically, $\boldsymbol{u'}$ and $\boldsymbol{v'}$ in $\hat{{\cal E}}_{k'}^r$ are modified from $\boldsymbol{u}$ and $\boldsymbol{v}$ in ${{\cal E}}_{k'}^r$ by appending $h(k)-h(k')$ number of zeros, where $h(k)$ is the height of node $k$ in the search tree.
The Benders cuts of other ancestor nodes are able to be inherited in the same way.
Then, the augmented master problem (AMP) of node $k$, which inherits Benders cuts from all the direct ancestor nodes in the search tree, can be constructed by combining the inherited cuts and the master problem.
After obtaining AMP and ARSP, the BD algorithm is used again and generates a new node to the search tree for subsequent iterations until the BDBC algorithm converges to the optimal solution.
The details of BDBC algorithm are provided in Algorithm \ref{alg:BDBC}.

\begin{lemma}
	\label{lm:lemma4}
	The BDBC algorithm is able to converge to the optimal solution of the primal problem within finite iterations.
\end{lemma}

\begin{proof}
	The BDBC algorithm is built by integrating the BD algorithm into the branch and bound algorithm.
	Specifically, a search tree for branch and bound algorithm is built and maintained, and each node of the tree performs BD algorithm.
	Since both the BD algorithm and branch-and-bound algorithm are able to converge after finite iterations \cite{cordeau2000bendersdecomposition, huang2019branchandcutbenders}, it is intuitive that the BDBC algorithm is also able to converge after finite iterations.
\end{proof}

\begin{figure}[tbp]
	\centering
	\includegraphics[width=0.48\textwidth]{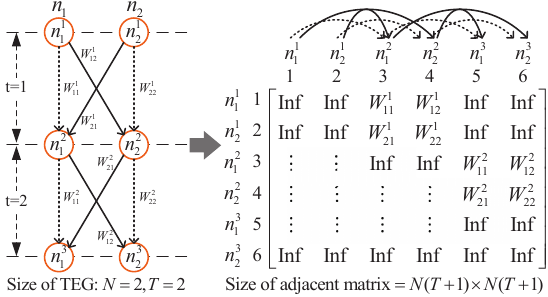}
	\caption{An example of adjacency matrix constructed from TEG model.}
	\label{fig:des_adjacentMatrix}
\end{figure}

\section{Time Expansion-based Heuristic Algorithm}
\label{sec:greedySolution}

\begin{algorithm}[!htbp]
	\begin{footnotesize}
		\DontPrintSemicolon
		\caption{TEDG algorithm}
		\label{alg:TEDG}
		\textbf{Initialization:} ${\cal G} = \left( {{\cal N},{\cal L}} \right)$ with ${{C_n}}$ for each network node $n$, $\varphi _{nm}^t$ and ${R_{nm}^t}$ for each link $nm$; Services $\cal Q$ with ${\delta _q}$, ${c_q}$, ${s_q}$, ${e_q}$, and ${{\cal I}_q}$ for each $q$\;

		\For{each $q$ in $\cal Q$}{
			$T_{min}^q \Leftarrow$ Calculate the minimum required computation time slots by (\ref{eq:minComputationTimeslot})\;
			Set $t = T_{min}^q$\;
			Set $serviceStatus=0$\;
			${id}_s^q \Leftarrow$ Get the index of source node of service $q$\;
			$n_e^q \Leftarrow$ Get the index of the destination node of service $q$\;
			/* \textit{Time expansion for service planning} */\;
			\While{true}{
				\If{$t>T$}{
					Discard service $q$\;
					Break\;
				}
				Set ${id}_e^q = n_e^q + N \times t$ and update the index of the destination node in the adjacency matrix\;
				${\rm A}_q^t \Leftarrow$ Construct the adjacency matrix with the length of time slots $t$ using Algorithm \ref{alg:adjacentMatrixConst}\;
				${\cal P}_q^t \Leftarrow$ Generate the potential path set according to ${\rm A}_q^t$ using Algorithm \ref{alg:pathSetGeneration}\;
				\If{${\cal P}_q^t = \emptyset$}{
					$t = t+1$\;
					Continue\;
				}
				\For{each path $p_q^t$ in ${\cal P}_q^t$}{
					$vnfFlag \Leftarrow$ Try to deploy all VNFs on $p_q^t$ using Algorithm \ref{alg:vnfDeployment}\;
					\If{$vnfFlag$}{
						$serviceStatus=1$\;
						Update the communication resource status of communication links in $p_q^t$\;
						Break\;
					}
				}
				\If{$serviceStatus$}{
					Break\;
				}
				\Else{
					$t = t+1$\;
				}
			}
		}
	\end{footnotesize}
\end{algorithm}

The BDBC algorithm is still computationally intensive due to the ILP master problem and the branch-and-bound tree.
To reduce computational complexity for practical use, we propose TEDG algorithm by considering the characteristics of the TEG model and the optimization objective of the primal problem.

\subsection{Details of TEDG Algorithm}
The detailed process of the TEDG algorithm is presented in Algorithm \ref{alg:TEDG}.
It is assumed that the dynamics of SDSTN topology in future $T$ time slots are known in advance at the start of Algorithm \ref{alg:TEDG}, due to the predictability of satellite movement trajectories.
In addition, all service requirements are collected by network controllers for on-demand scheduling.
Note that Algorithm \ref{alg:TEDG} decouples routing planning and VNF placement into two stages, where lines 13-18 generate potential service flow paths, and lines 19-24 perform VNF placement.
For the path search, the TEG needs to be converted into an adjacency matrix without losing the dynamic characteristics of network topology under resource constraints.
The adjacency matrix describes the TEG of $T$ time slots, where each dimension represents the index of all the nodes in the TEG, given by $[1, 2, \dots, N, \dots, N \times (T+1)]$.
For example, in each dimension of the adjacency matrix, the index of node $n \in {\cal N}$ at time slot $t$ in the TEG is $(n+N \times (t-1))$.
Fig. \ref{fig:des_adjacentMatrix} shows the transformation of a two-slot TEG into an adjacency matrix, where each weight represents the corresponding link cost.
In particular, line 2 of Algorithm \ref{alg:adjacentMatrixConst} indicates that except for feasible links in the TEG, the cost of the other links is infinite.

\begin{algorithm}[tbp]
	\begin{footnotesize}
		\caption{Adjacency matrix construction}
		\label{alg:adjacentMatrixConst}
		\DontPrintSemicolon
		\SetKwInOut{Input}{Input}
		\SetKwInOut{Output}{Output}

		\Input{The TEG of the network ${\cal G} = \left( {{\cal N},{\cal L}} \right)$, $t$, $\tau$, information of service $q$}
		\Output{${\rm A}_q^t$}
		\BlankLine
		Set $S_q^t = N \times (t + 1) $ and get the size of adjacency matrix ${\rm A}_q^t$ based on the length of time slots $t$\;
		Initialize each element of ${\rm A}_q^t$ as $\infty$\;
		\For{$t' = 1:t$}{
		\For{each link $nm$ in ${\cal L}$}{
		$n^{t'}_{id}=(t' - 1) \times N + n$ and set the index of node $n$ at the ending of time slot $t'-1$\;
		$m^{t'}_{id}=t' \times N + m$ and set the index of node $m$ at the ending of time slot $t'$\;

		\If{$\varphi_{nm}^{t'} = 1$ and $n = m$}{
		\If{\begingroup\makeatletter\def\f@size{6.5}\check@mathfonts
		${{min(\{C_n^{t'}|{t'\le t}\}) \ge (min(\{\beta_i^q|i \in {\cal I}_q\}) \times \tau + c_q)}}$\endgroup}{
		\If{\begingroup\makeatletter\def\f@size{6.5}\check@mathfonts
		${{min(\{C_n^{t'}|{t'\le t}\}) = max(\{C_n^{t'}|{t'\le t}\})}}$\endgroup}{
		${\rm A}_q^t(n^{t'}_{id},m^{t'}_{id}) = 0.5$\;
		}
		\Else{
		${\rm A}_q^t(n^{t'}_{id},m^{t'}_{id}) = {W^{t'}_{{nn}}}$\;
		}
		}
		\Else{
		${\rm A}_q^t(n^{t'}_{id},m^{t'}_{id}) = 0.9$\;
		}
		}
		\ElseIf{$\varphi_{nm}^{t'} = 1$ and $n \ne m$}{
		\If{$\frac{\delta_q}{R_{nm}^{t'}} \le \tau$}{
		${\rm A}_q^t(n^{t'}_{id},m^{t'}_{id}) = 1$\;
		}
		\Else{
		${\rm A}_q^t(n^{t'}_{id},m^{t'}_{id}) = \infty$\;
		}
		}
		}
		}
	\end{footnotesize}
\end{algorithm}

As a prerequisite for the path search in Algorithm \ref{alg:pathSetGeneration}, it is necessary to normalize the cost of links.
For communication links with sufficient communication resource for data transmission of a given service, the cost is considered uniform and set to 1 in the adjacency matrix; otherwise, the cost is set as infinite.
Regarding stay links, those capable of only store and forward operations have a cost of 0.9, while those with sufficient computation resource have a cost of 0.5.
If the aforementioned conditions are not met, the cost of stay links is normalized using the modified max-min method, which is given by
\begin{equation}
	\label{eq:mmMethod4Staylink}
	W^t_{nn} = {0.9 - \frac{{0.4(C_n^t - \min (\{ C_n^t|t \in {\cal T}\} ))}}{{\max (\{ C_n^t|t \in {\cal T}\} ) - \min (\{ C_n^t|t \in {\cal T}\} )}}},
\end{equation}
where $C_n^t$ denotes the available computation resource of node $n$ in time slot $t$.

Following the construction of the adjacency matrix, Algorithm \ref{alg:pathSetGeneration} is used to generate a set of potential paths, ${\cal P}_q^t$.
Specifically, line 1 of Algorithm \ref{alg:pathSetGeneration} identifies the $k$ shortest paths between the source and destination nodes, which are represented as the set ${\cal R}_q^t$.
To generate a robust path set, Yen's $k$-shortest path algorithm \cite{yen1971findingk} is utilized.
Subsequently, for each path in ${\cal R}_q^t$, the algorithm calculates the number of time slots occupied by all stay links, determining the path's feasibility based on whether the time spent on stay links exceeds the time required for all VNFs to process the service data.
Paths that do not satisfy this condition are omitted.
Lines 2-6 in Algorithm \ref{alg:pathSetGeneration} illustrate the aforementioned operations.

\begin{algorithm}[tbp]
	\begin{footnotesize}
		\DontPrintSemicolon
		\caption{Potential path set generation}
		\label{alg:pathSetGeneration}
		\SetKwInOut{Input}{Input}
		\SetKwInOut{Output}{Output}

		\Input{${\rm A}_q^t, {id}_s^q, {id}_e^q, T_{min}^q$}
		\Output{${\cal P}_q^t$}

		\BlankLine
		${\cal R}_q^t \Leftarrow$ Generate the initial path set using $k$-shortest path algorithm\;

		\If{${\cal R}_q^t \ne \emptyset$}{
			\For{each path $r_q^t$ in ${\cal R}_q^t$}{
				$T^{'} \Leftarrow$ Calculate the total number of time slots occupied by all stay links in $r_q^t$\;
				\If{$T^{'} \ge T_{min}^q$}{
					Add path $r_q^t$ into ${\cal P}_q^t$\;
				}
			}
		}
	\end{footnotesize}
\end{algorithm}

\begin{algorithm}[bp]
	\begin{footnotesize}
		\caption{VNF placement}
		\label{alg:vnfDeployment}
		\SetKwInOut{Input}{Input}
		\SetKwInOut{Output}{Output}
		\DontPrintSemicolon
		\Input{The TEG ${\cal G} = \left( {{\cal N},{\cal L}} \right)$, $p_q^t$, ${id}_s^q$, ${id}_e^q$}
		\Output{$vnfFlag$}

		\BlankLine
		Place the first and last VNF on node ${id}_s^q$ and ${id}_e^q$\;
		$vnfFlag = 0$\;
		Set the index of the next VNF to be placed as $i = 2$\;
		${\cal L}'_{s} \Leftarrow$ Find all stay links in path $p_q^t$\;

		\For{each stay link $l'_{s}$ in ${\cal L}'_{s}$}{
		\If{$i > (I_q-1)$}{
			Break\;
		}
		\While{$i \le (I_q-1)$}{
		\If{Node of $l'_{s}$ has enough computation resources for deploying VNF $i$ and processing service data}{
		Place VNF $i$ on this node\;
		$i = i+1$\;
		}
		\Else{
			Break\;
		}
		}
		}

		\If{$i > (I_q-1)$}{
			Update the computation resource status of nodes in $p_q^t$\;
			\Return $vnfFlag = 1$\;
		}
		\Else{
			\Return $vnfFlag = 0$\;
		}
	\end{footnotesize}
\end{algorithm}

In Algorithm \ref{alg:TEDG}, path searching and VNF placement start by considering the minimum slot requirements $T_{min}^q$ of each service, which is given by
\begin{equation}
	\label{eq:minComputationTimeslot}
	T_{min}^q = \left\lceil {\sum\limits_{i \in {{\cal I}_q}} {\frac{{ {\delta _q}\varepsilon}}{{{c_q}\tau }}} } \right\rceil,
\end{equation}
where $\lceil \cdot \rceil $ represents the ceiling function, which rounds up its argument to the nearest integer.
For a given service $q$ with destination node $m$ and $t=T_{min}^q$, the index of this destination node in the adjacency matrix is $m + NT_{min}^q$.
Next, Algorithm \ref{alg:pathSetGeneration} identifies all potential paths from the service source node to this index.
If there is no feasible path spanning $T_{min}^q$ time slots in the TEG, the index of the destination node in the adjacency matrix is further extended to the next time slot with the index of $m + N(T_{min}^q+1)$, until a feasible path is discovered.
If the index of the destination node reaches $m + N(T+1)$ without finding a feasible path, service $q$ is disregarded.

After obtaining a non-empty set of potential paths, VNF deployment is conducted, as shown in lines 19-24 of Algorithm \ref{alg:TEDG} and detailed in Algorithm \ref{alg:vnfDeployment}.
First, the stay links and their respective time slot lengths within the path are determined.
Then, VNFs are deployed in sequence based on the stay link order and resource status.
The VNF placement for a single service is deemed complete once all required VNFs are successfully deployed.
Consequently, the network status, including available communication and computation resource, is updated, and planning for the next service proceeds until all services are addressed.

\subsection{Algorithm Complexity}
\label{sec:algorithmCom}
TEDG algorithm involves adjacency matrix construction, potential path set generation, and VNF placement.
Specifically, the complexity of the adjacency matrix construction is determined by Algorithm \ref{alg:adjacentMatrixConst}, which is $\mathcal{O}(tN^2)$.
The complexity of the potential path set generation is dominated by $k$-shortest path searching in Algorithm \ref{alg:pathSetGeneration}, whose complexity is $\mathcal{O}(k(E+S\log S))$ \cite{berclaz2011multipleobject}, where $S =tN+N$ is the length of one dimension in the adjacency matrix, and $E=tN^2$ is the number of available links in the TEG with $t$ time slots.
Algorithm \ref{alg:vnfDeployment} is used to find the feasible path from a path set to place all the VNFs of the service, and the complexity is determined by the number of time slots occupied by stay links and the number of required VNFs.
In the worst case, there is only one communication link, and the other time slots are occupied by stay links.
At this time, the complexity of VNF placement on one path is $\mathcal{O}(I_q(t-1))$.
Since there are at most $k$ paths in ${\cal P}_q^t$, the complexity of Algorithm \ref{alg:vnfDeployment} is $\mathcal{O}(kI_q(t-1))$ in the worst case.
Therefore, the overall complexity of the proposed TEDG algorithm is $\mathcal{O}(Q\sum_{t=1}^{T} (tN^2+ktN^2+k(tN+N)\log (tN+N)+kI_q(t-1))) \cong \mathcal{O}(\frac{QTN^2(1+k)(1+T)}{2})$.

\begin{table}[tbp]
	\centering
	\caption{Simulation Parameters}
	\label{tab:simulationParas}
	\begin{tabular}{c|c}
		\toprule[1pt]
		\textbf{Parameter}                       & \textbf{Value}         \\ \midrule[1pt]
		The number of satellites                 & 12                     \\
		The number of ground stations            & 4                      \\
		The number of ground users               & 4                      \\
		The number of orbits                     & 4                      \\
		Orbits altitude                          & 700 $km$                 \\
		Orbits inclination                       & 45 deg                 \\
		Configuration period                     & 1 hour                 \\
		Time slot length                         & 100 s                  \\
		Service data size                        & $[100, 1000]$ Mbit     \\
		Service computation resource requirement & \{20,80\} unit/100Mbit \\
		VNF computation resource requirement     & \{20,30\} unit         \\
		The number of VNFs for each service      & 4                      \\
		Bandwidth                                & 20 MHz                 \\
		Carrier frequency                        & 20 GHz                 \\
		Thermal noise                            & 354.81 K               \\
		EIRP plus receiver antenna gain          & 3.74 W                 \\
		\bottomrule[1pt]
	\end{tabular}
\end{table}

\section{Performance Evaluation}
\label{sec:performanceEvaluations}
In this section, extensive experiments are conducted to evaluate the performance of the proposed algorithms.
The Walker constellation is used to construct an LEO satellite network in AGI Systems Tool Kit (STK) \cite{agi2022systemstool}, based on which the parameters of the dynamic SDSTN topology are generated.
All algorithms are executed on a simulation platform with an AMD Ryzen 5 5600X processor and 16GB of memory, and GUROBI is used to solve the master problem and subproblem in the BDBC algorithm.

\begin{figure}[tbp]
	\centering
	\begin{minipage}[t]{0.8\linewidth}
		\centering
		\includegraphics[width=\textwidth]{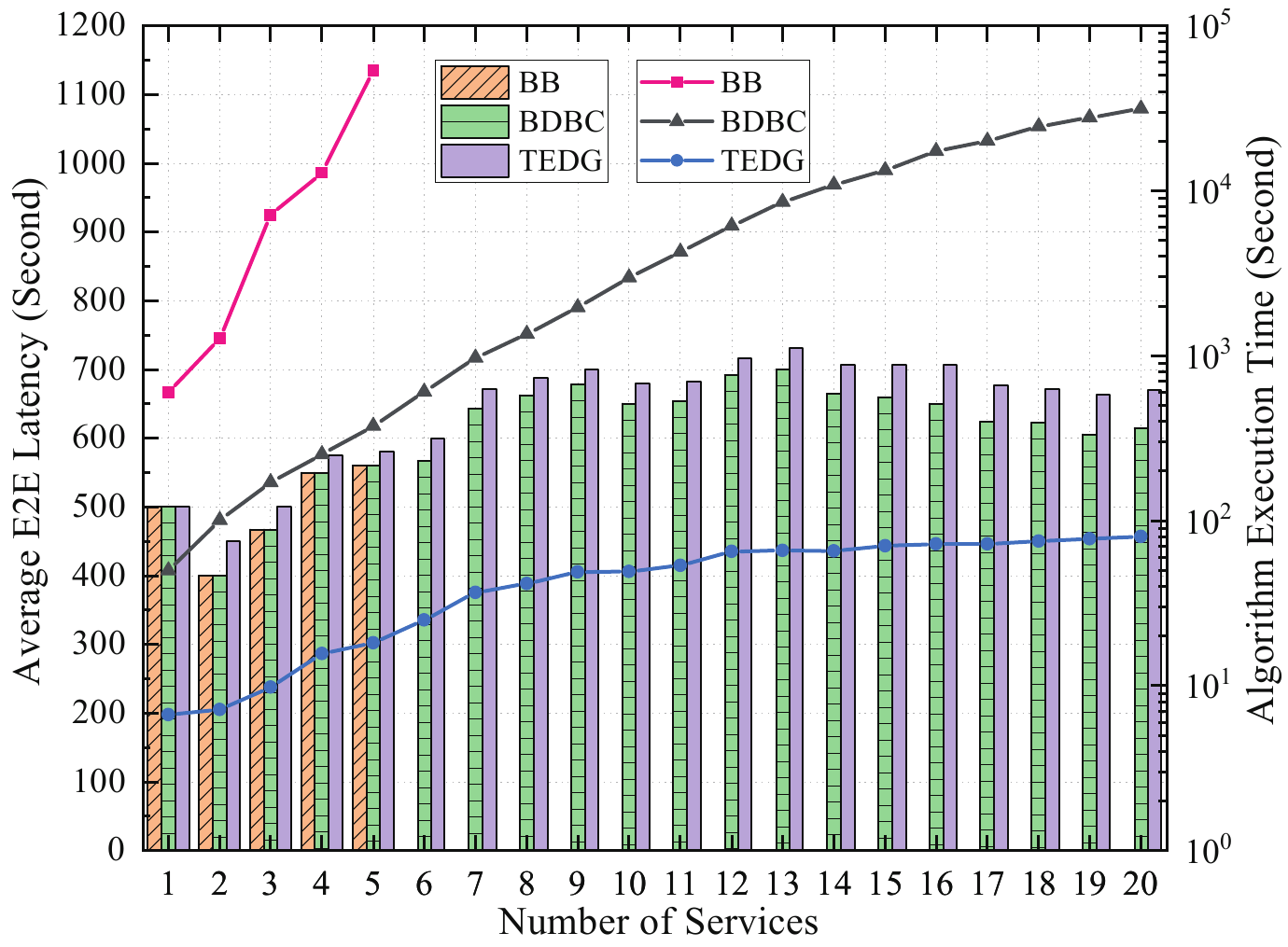}
		\vspace{-0.25in}
		\caption{Performance of different algorithms compared with the optimal solution. The bars indicate average service latency and the lines indicate algorithm execution time.}
		\label{fig:eva_algsvstaskNum}
	\end{minipage}%
	\vspace{0.15in}
	\hfill%
	\begin{minipage}[t]{0.8\linewidth}
		\centering
		\includegraphics[width=\textwidth]{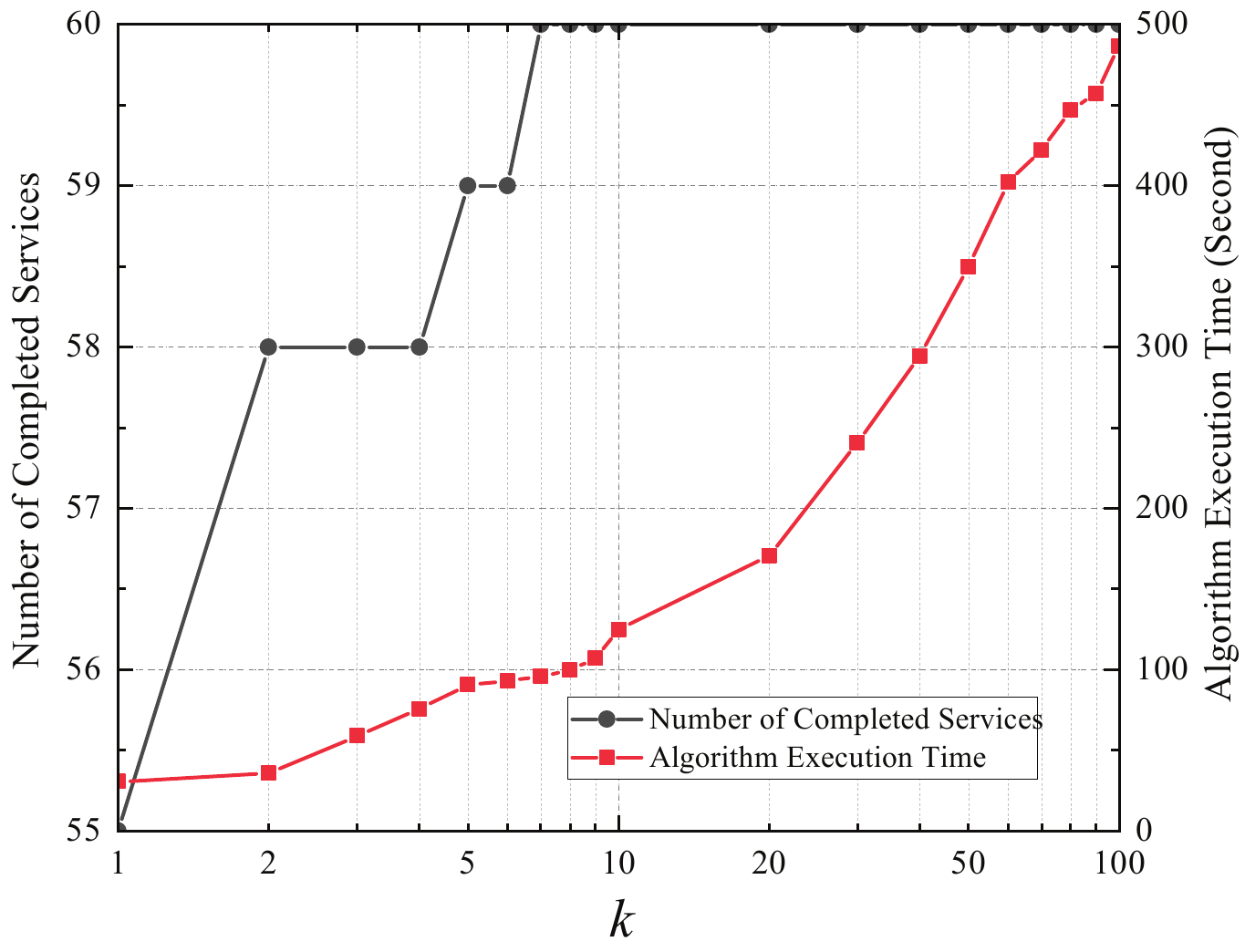}
		\vspace{-0.25in}
		\caption{The number of completed services and algorithm execution time versus the number of shortest paths $k$ in Algorithm \ref{alg:pathSetGeneration}.}
		\label{fig:eva_kvscompletedService}
	\end{minipage}
\end{figure}

\begin{figure*}[tbp]
	\centering
	\subfigcapskip=-5pt
	\subfigure[]{
		\label{fig:eva_ccPerformance20q}
		\includegraphics[width=0.43\linewidth]{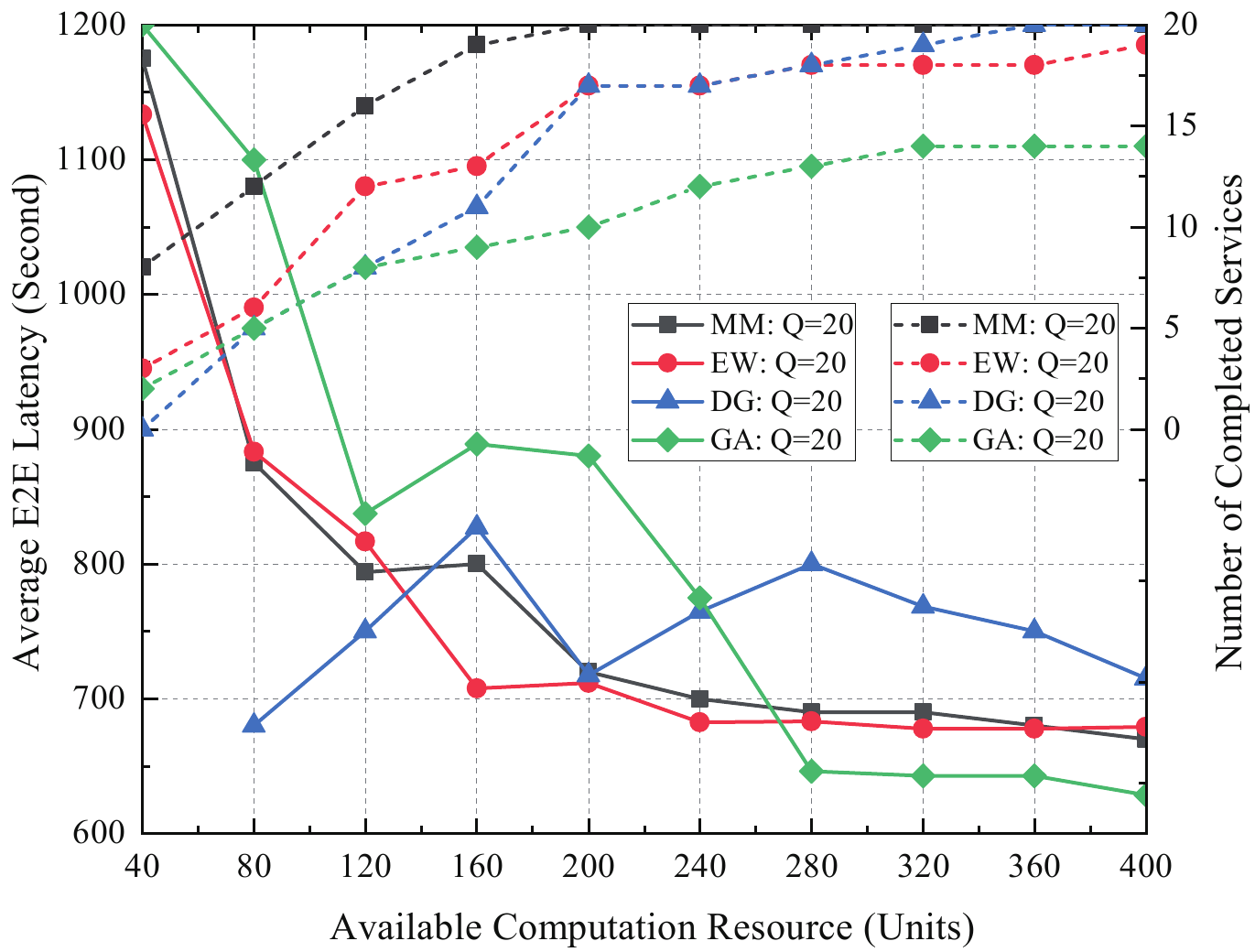}}
	\subfigure[]{
		\label{fig:eva_ccPerformance40q}
		\includegraphics[width=0.43\linewidth]{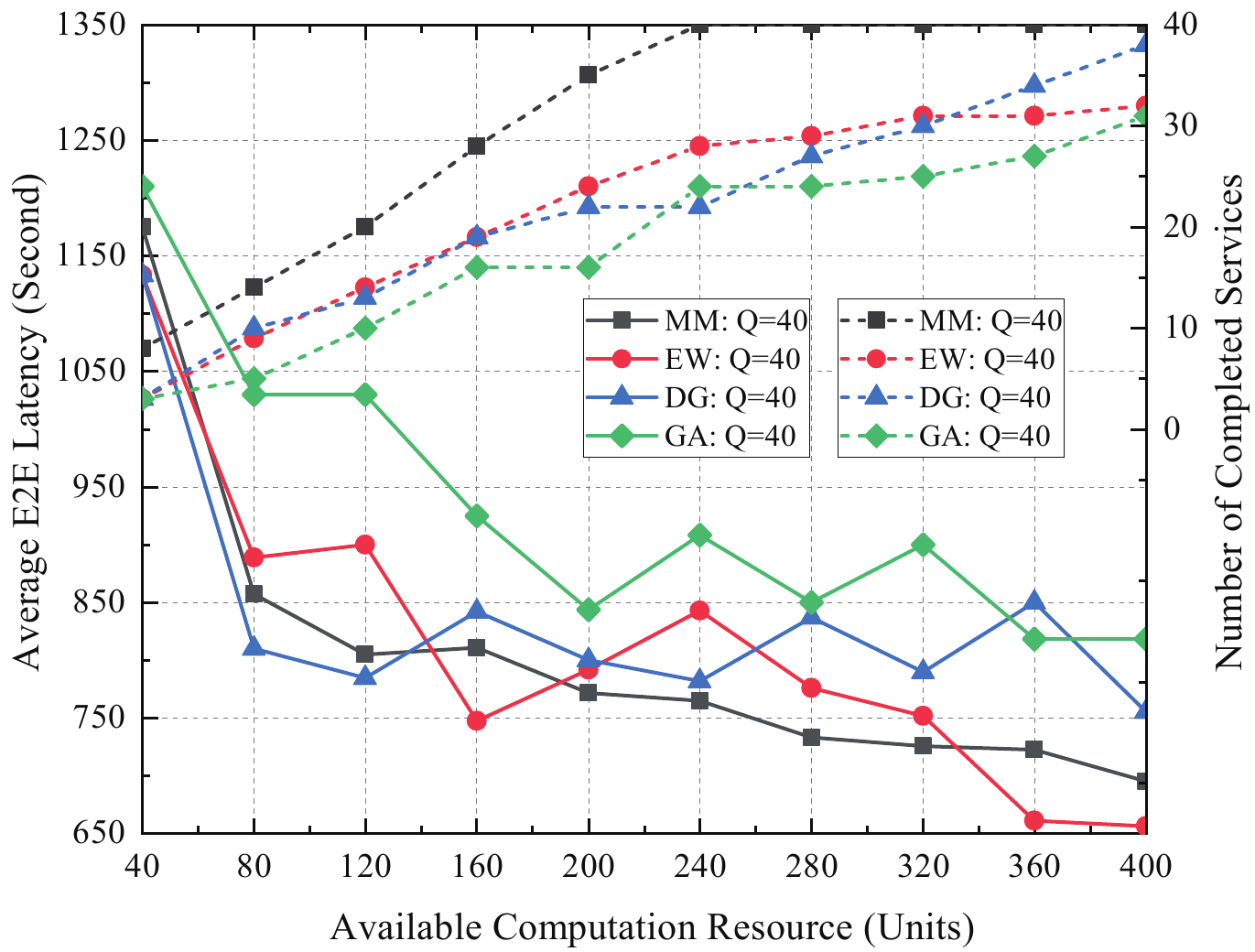}}

	\subfigure[]{
		\label{fig:eva_ccPerformance60q}
		\includegraphics[width=0.43\linewidth]{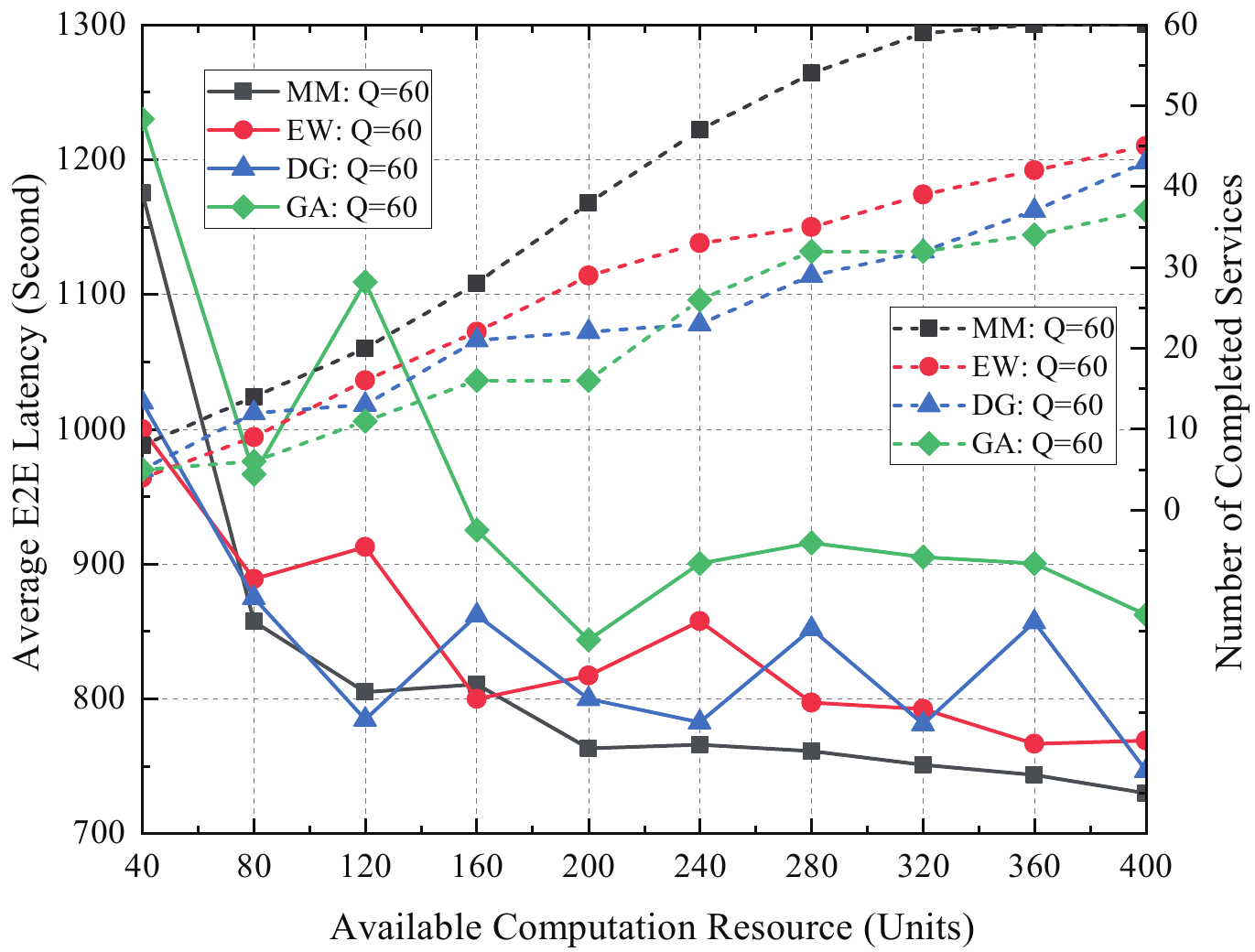}}
	\subfigure[]{
		\label{fig:eva_ccPerformanceMM}
		\includegraphics[width=0.43\linewidth]{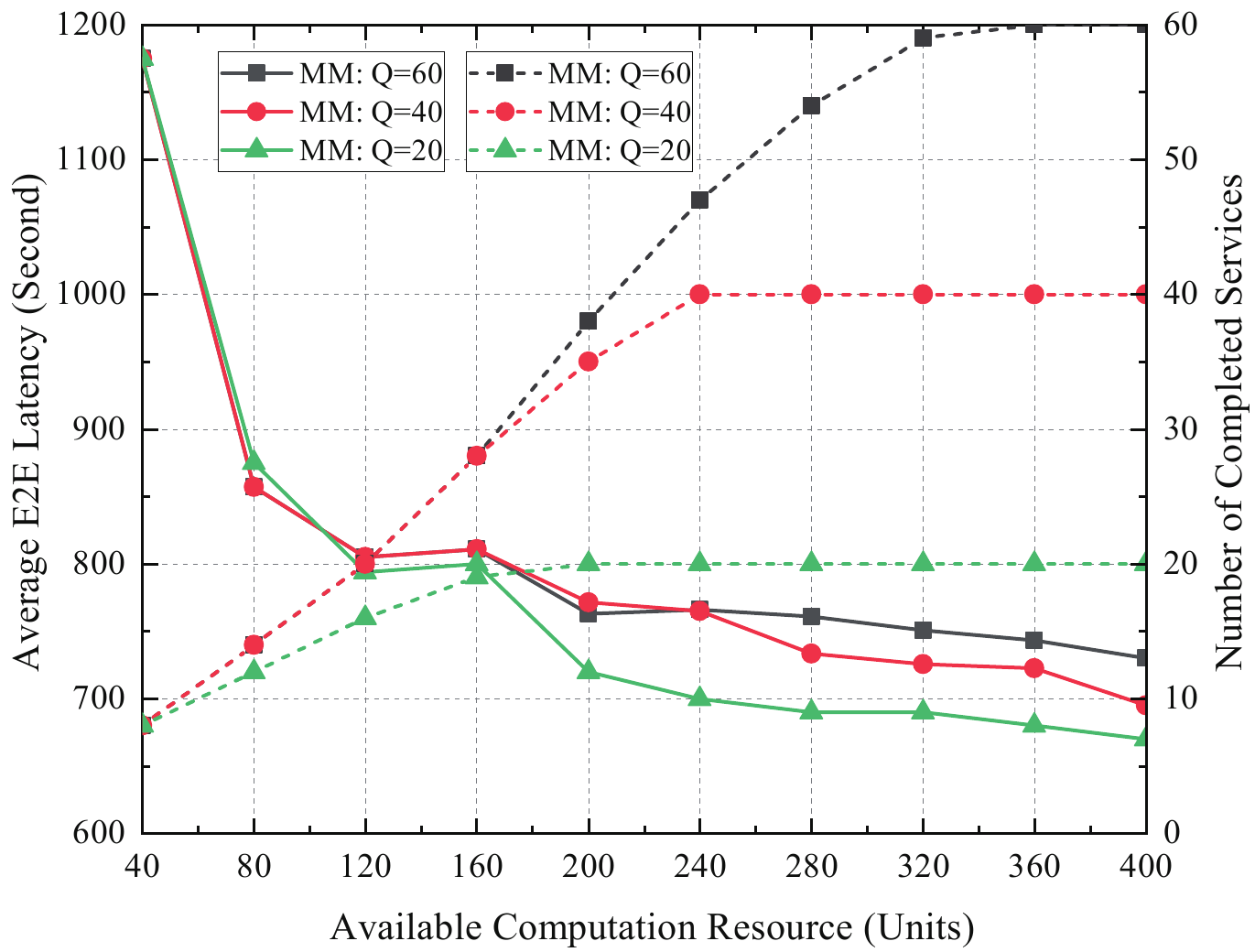}}
	\caption{Impact of available computation resource per node on the network performance (12 software-defined satellites, 4 ground stations, and 4 ground users). The solid lines correspond to the average service latency, while the dashed lines correspond to the number of completed services.}
	\label{fig:eva_ccPerformance}
\end{figure*}

\subsection{Simulation Setup}
The common simulation parameters for all experiments are presented in Table \ref{tab:simulationParas}.
The SDSTN scenario considered in our experiments includes 12 satellites which are distributed among 4 orbits with ${45^ \circ }$ inclination and altitude of 700 $km$, and the orbit period is 5927 seconds.
In addition, there are 4 ground stations located at Jiuquan (39.76$^\circ$N, 98.56$^\circ$E), Taiyuan (37.87$^\circ$N, 112.56$^\circ$E), Wenchang (19.62$^\circ$N, 110.75$^\circ$E), and Xichang (27.89$^\circ$N, 102.27$^\circ$E), respectively, and 4 ground users located at (31.49$^\circ$N, 110.13$^\circ$E), (34.45$^\circ$N, 84.98$^\circ$E), (52.26$^\circ$N, 124.35$^\circ$E), and (21.98$^\circ$N, 100.94$^\circ$E), respectively.
The configuration period is set to 3600 seconds, consequently, the network topology is generated by STK over a time horizon of 3600 seconds, spanning from 10 April 2022 06:00:00.000 UTCG to 10 April 2022 07:00:00.000 UTCG.
Moreover, the number of time slots is set as $T=36$, and thus each time slot has a length of $\tau=100$ seconds.
Communication links operate in Ka-band with a bandwidth of 20 MHz.
The transmission capacity between two ground stations is set as 1 Gbps and all ground stations are able to place VNFs.
We assume that each service requires two VNFs in addition to the first and last dummy VNF, and hence we have $I_q = 4$.
The computation resource required for hosting each VNF is either 20 units or 30 units.
The source node and destination node of each service are randomly selected according to the service type.

\subsection{Trade-off Between Complexity and Performance}
To verify the optimality of the proposed BDBC algorithm, we compare it with the solution of the branch-and-bound (BB) algorithm.
However, due to the high computational complexity of both BB and BDBC, the corresponding comparison is limited to a maximum of 20 services.
In addition, the performance of the proposed TEDG algorithm is also compared with these optimal algorithms to demonstrate the balance achieved by the TEDG algorithm between complexity and performance.
Furthermore, we conduct an evaluation to investigate the appropriate value for $k$, the number of shortest paths generated by Algorithm \ref{alg:pathSetGeneration}, to achieve a balance between the performance and execution time of the TEDG algorithm.

\begin{figure*}[tbp]
	\centering
	\subfigure[]{
		\label{fig:eva_vSatPerformance20q}
		\includegraphics[width=0.43\linewidth]{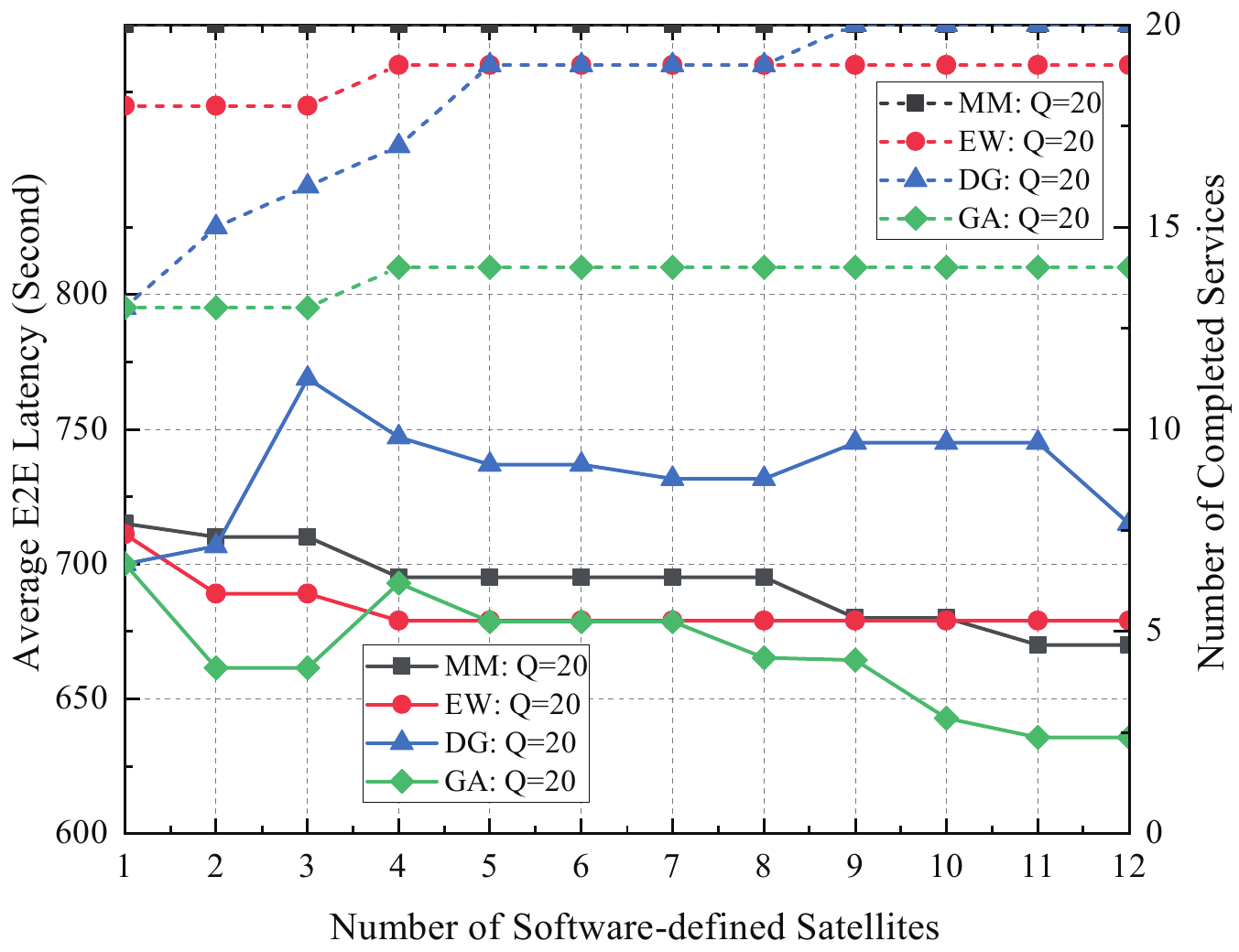}}
	\subfigure[]{
		\label{fig:eva_vSatPerformance40q}
		\includegraphics[width=0.43\linewidth]{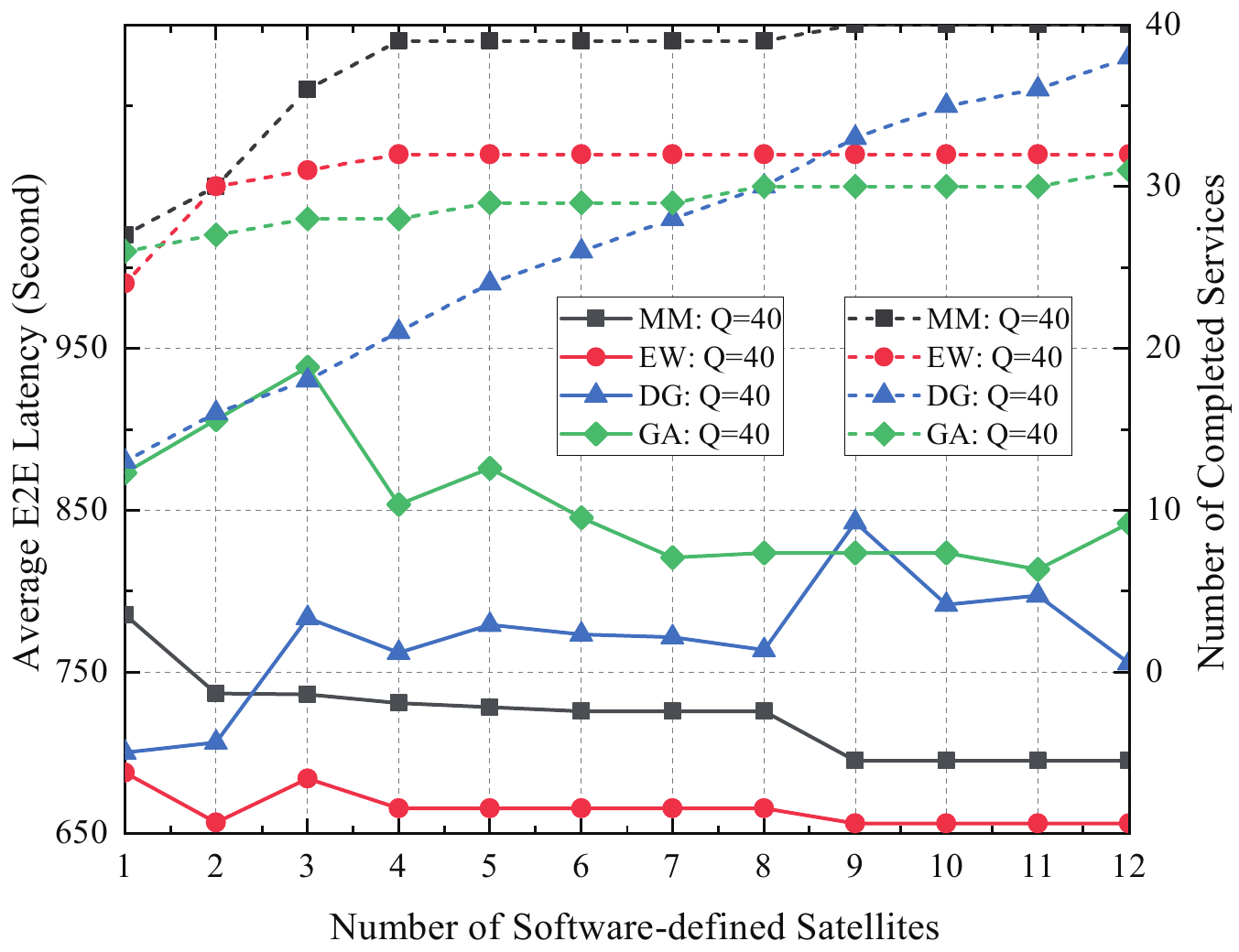}}

	\subfigure[]{
		\label{fig:eva_vSatPerformance60q}
		\includegraphics[width=0.43\linewidth]{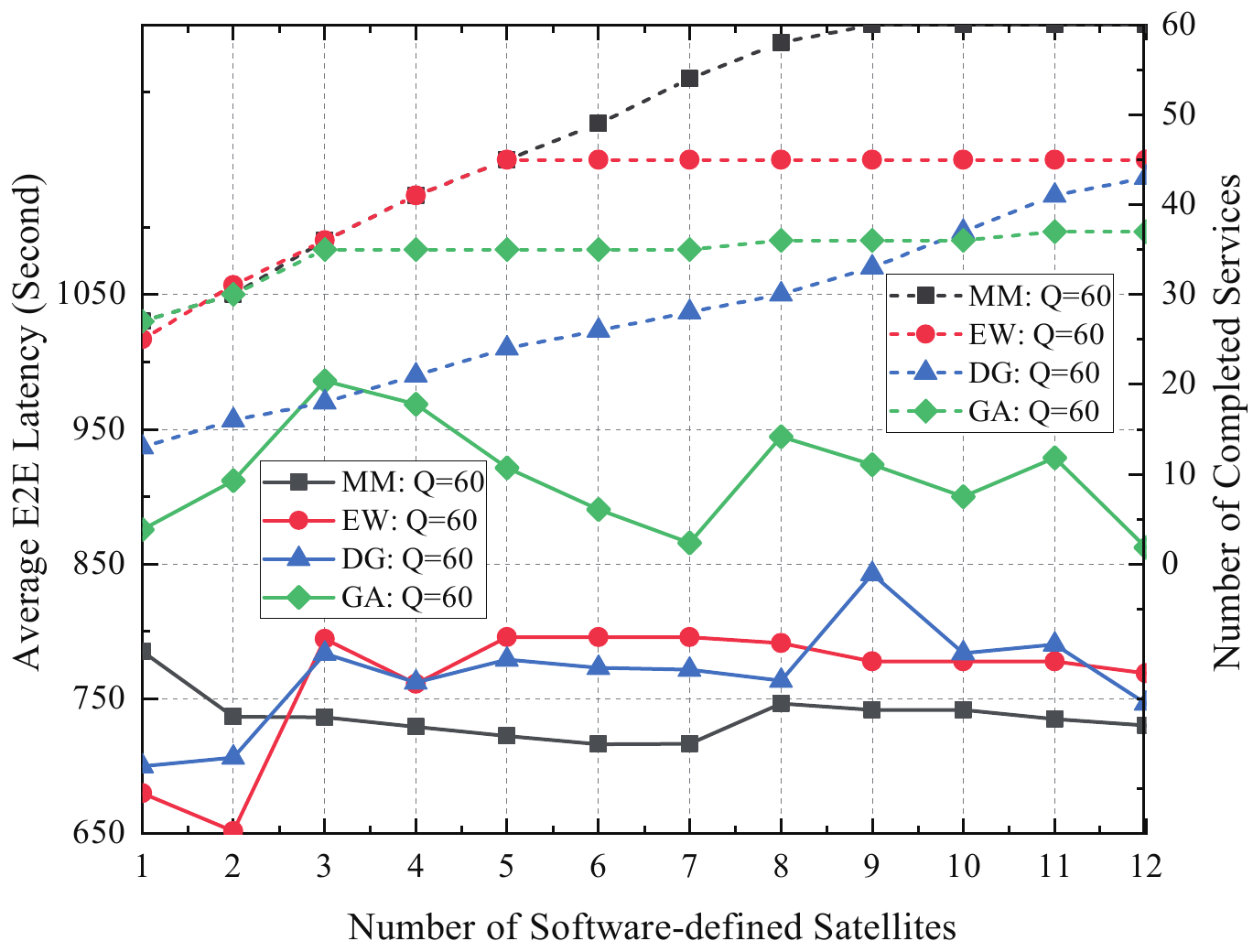}}
	\subfigure[]{
		\label{fig:eva_vSatPerformanceMM}
		\includegraphics[width=0.43\linewidth]{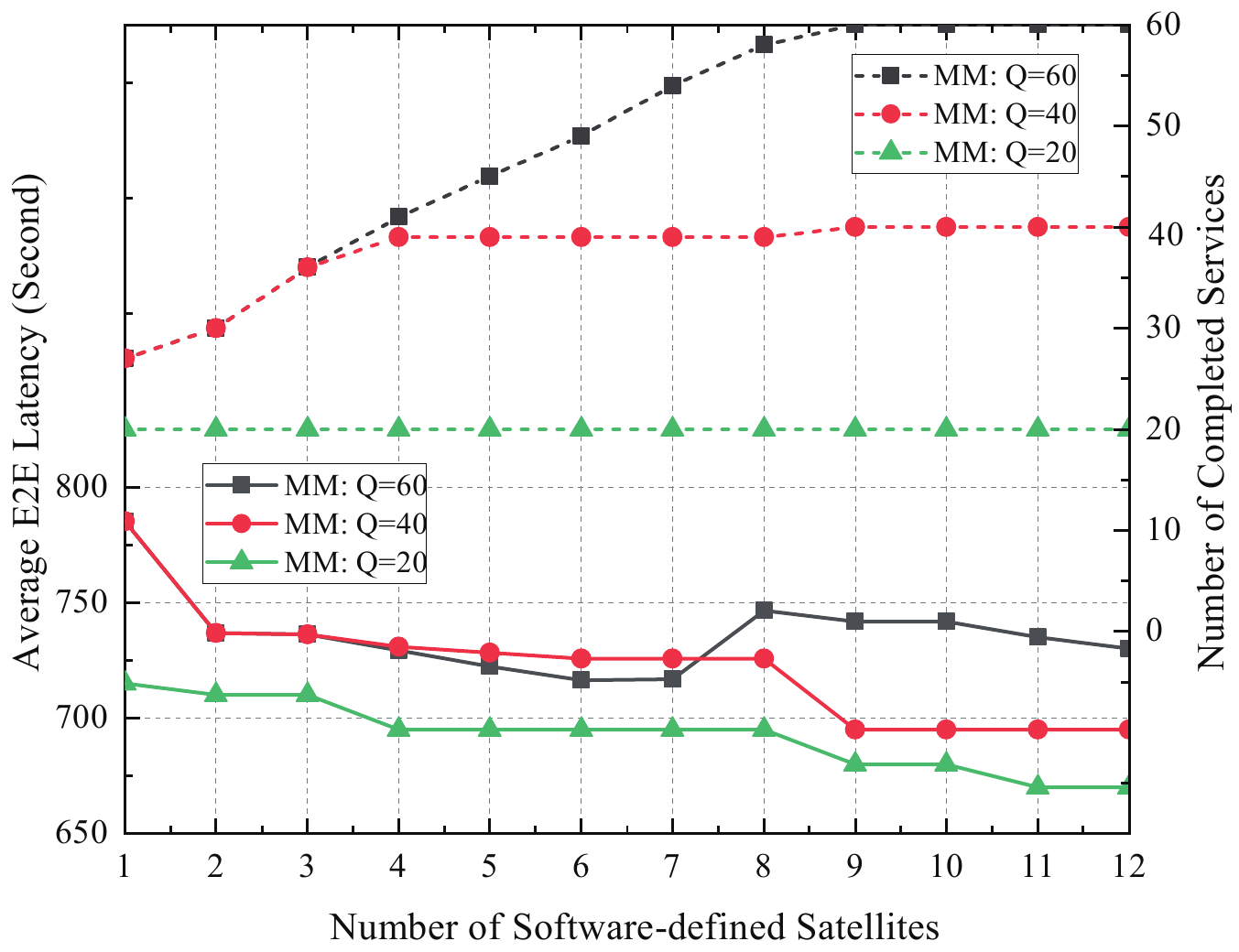}}
	\caption{Impact of the number of software-defined satellites on the network performance (12 satellites, 4 ground stations, 4 ground users, and 400 units computation resource per VNF-enabled node). The solid lines correspond to the average service latency, while the dashed lines correspond to the number of completed services.}
	\label{fig:eva_vSatPerformance}
\end{figure*}

In Fig. \ref{fig:eva_algsvstaskNum}, the optimization objective of our formulated problem, ``Average E2E Latency", is represented by bars and serves as a key metric to evaluate the effectiveness of the proposed algorithms.
Additionally, the metric ``Algorithm Execution Time" is represented by solid lines in Fig. \ref{fig:eva_algsvstaskNum} and is used to assess the time complexity of algorithm execution.
In the experiment, the number of software-defined satellites is set to 12, and each network node has 400 units of computation capacity.
Meanwhile, the number of shortest paths $k$ in Algorithm \ref{alg:pathSetGeneration} of TEDG is set to 100.
Moreover, since the BB algorithm is very time-consuming, we only evaluate its performance with a few services.
Based on the results, the optimality of our proposed BDBC algorithm is verified by comparing it with the BB algorithm.
However, the computational complexity of the BDBC algorithm is relatively lower compared to the BB algorithm, owing to the search space being tightened by Benders cuts iteratively.
For the TEDG algorithm, even though it does not provide an optimal solution, it approaches optimal performance while substantially reducing computational complexity compared to both BB and BDBC algorithms.
Specifically, when there are 20 services, the execution time of TEDG algorithm is significantly reduced by nearly three orders of magnitude compared to BDBC algorithm, even though the average service latency performance of the TEDG algorithm is reduced by 7\%.

Next, we evaluate the impact of the number of shortest paths $k$ in Algorithm \ref{alg:pathSetGeneration} on the performance of the TEDG algorithm.
In the experiment, we consider a scenario with 12 software-defined satellites, each equipped with 400 units of computation resource.
Additionally, a total of 60 services require planning.
In Fig. \ref{fig:eva_kvscompletedService}, the metric ``Number of Completed Services" represents the number of services successfully delivered from the source node to the destination node within the configuration period and is used to evaluate the performance of the TEDG algorithm, while the ``Algorithm Execution Time" metric is used to estimate the time cost of TEDG.
The figure shows that as $k$ increases, the number of completed services gradually increases until all services are completed.
This is reasonable, a larger value of $k$ provides more path choices for subsequent VNF deployment.
However, excessive path searching may not contribute to performance improvement when $k$ is sufficiently large.
Moreover, an increase in $k$ results in a longer algorithm execution time, which is basically consistent with the computational complexity analysis presented in Section \ref{sec:algorithmCom}.

\subsection{Performance Evaluation of TEDG Algorithm}

To demonstrate the efficiency of our proposed TEDG algorithm with the modified max-min (MM) normalization method as outlined in Section \ref{sec:greedySolution}, we consider the following three schemes as baselines:
\begin{enumerate}
	\item \emph{TEDG algorithm with equal weight (EW) scheme:}
	      We employ a variation of the TEDG algorithm where the cost of all available links, including communication and stay links, is assigned an equal weight of 1.
	\item \emph{Decoupled greedy algorithm (DG) scheme:}
	      The decoupled greedy algorithm proposed in \cite{wang2020sfcbasedservice} is modified to fit our scenario, which allocates resource to services on a configuration period basis instead of a time slot basis.
	\item \emph{Genetic algorithm (GA) scheme:}
	      GA is applied to solve the joint VNF placement and routing planning problem \eqref{eq:primalProblem}, and it has also been adopted in \cite{rankothge2017optimizingresource}.
\end{enumerate}

We investigate the effectiveness of TEDG in joint VNF placement and routing planning, examining the impact of the number of services, the onboard computation capacity of each software-defined satellite, and the number of software-defined satellites.
In this simulation part, we employ two key metrics for evaluation: ``Average E2E Latency", which is the optimization objective of our formulated problem, and ``Number of Completed Services", denoting the number of services successfully transmitted from the source node to the destination node within the configuration period.

The performance of different algorithms under different available computation resource is shown in Fig. \ref{fig:eva_ccPerformance20q}, Fig. \ref{fig:eva_ccPerformance40q}, and Fig. \ref{fig:eva_ccPerformance60q} for $Q=20$, $Q=40$, and $Q=60$, respectively.
In addition, the performance of TEDG algorithm with MM method is illustrated in Fig. \ref{fig:eva_ccPerformanceMM}.
It is evident that the number of completed services gradually increases with the increase of the computation capacity per node, as less stringent resource constraints result in a more extensive feasible solution space.
Moreover, the TEDG algorithm, whether employing the MM or EW method, outperforms both the DG and GA schemes in terms of completed services within the configuration period.
For instance, when $Q=60$ and the available computation resource at each network node is 400 units, TEDG algorithm with MM method is able to complete all 60 services and achieves a nearly 58\% improvement in terms of the number of completed services compared to the GA scheme, as well as a 43\% improvement compared to the DG scheme.
However, for the other three baselines, some services cannot be completed under the same computation resource configuration due to their lower resource utilization.
On the other hand, the average service latency tends to increase as the computation resource at each node decreases, as service flows must traverse more links to reach nodes with available computation resource.
In comparison to the three baselines, when $Q=60$, the TEDG algorithm with MM method achieves better latency performance while still completing a higher number of services.
Specifically, with 360 units of computation resource per network node, TEDG algorithm with MM method significantly reduces average service latency by roughly 17\% compared to GA scheme and 12\% compared to DG scheme.

Considering that not all satellites support VNF deployment in practice, we further investigate the impact of the number of software-defined satellites on both the number of completed services and the average service latency, as depicted in Fig. \ref{fig:eva_vSatPerformance}.
It can be seen that with the increase in the number of software-defined satellites, the number of completed services within the configuration period also increases.
Specifically, for TEDG algorithm with MM method, 9 software-defined satellites are required to complete all 60 services, and only 1 software-defined satellite is required to complete 20 services.
However, for the three baselines, namely TEDG algorithm with EW method, DG scheme, and GA scheme, up to 45, 43, and 37 services are completed, respectively, when 12 software-defined satellites are configured.
In fact, increasing the number of software-defined satellites implies an increase in the total computation resource of the network, subsequently reducing the communication hops between network nodes for services to reach the available computation resource.
In addition, Fig. \ref{fig:eva_vSatPerformance} clearly shows that the average service latency tends to decrease with the increase in the number of software-defined satellites.
Specifically, for TEDG algorithm with MM method and $Q=20$, since all services can be successfully planned when only one satellite supports VNF, the average service latency gradually decreases as the number of software-defined satellites increases.

\section{Conclusion}
\label{sec:conclusion}
In this paper, the joint virtual network function placement and routing planning in software-defined satellite-terrestrial networks has been investigated, aiming at minimizing the average end-to-end service latency, and time-evolving graph has been involved to characterize the dynamics of network topology and resource availability.
To deal with the formulated NP-hard problem, objective and constraints have been linearized and a Benders decomposition based branch-and-cut (BDBC) algorithm has been proposed.
Towards practical use, a time expansion-based decoupled greedy (TEDG) algorithm has been designed with rigorous complexity analysis.
Simulation results have revealed the optimality of BDBC algorithm and the low complexity of TEDG algorithm.

\ifCLASSOPTIONcaptionsoff
	\newpage
\fi

\bibliographystyle{IEEEtran}
\bibliography{bibUsedinPaper,bstControlForIEEEtran}

\end{document}